\documentclass[11pt]{article}

\usepackage{amsmath}
\usepackage{amsthm}
\usepackage{amsfonts}
\usepackage{mathtools}
\usepackage[english]{babel}
\usepackage{hyperref}
\usepackage{authblk}
\allowdisplaybreaks
\usepackage[a4paper, margin=1in]{geometry}
\usepackage{setspace}
\usepackage{mathtools}
\usepackage{graphicx}
\usepackage{epstopdf}
\usepackage{subcaption}
\usepackage{float}

\newtheorem{theorem}{Theorem}
\newtheorem{lemma}{Lemma}

\newcommand{\distortion}[1]{\psi_\gamma\left( #1 \right)}

\newcommand{\cdf}[2]{F_{#1}^\mathbb{Q}\left( #2 \right)}
\newcommand{\mubar}{\bar{\mu}}
\newcommand{\sigmabar}{\bar{\sigma}}
\newcommand{\bid}[1]{\mathrm{bid}\left( #1 \right)}
\newcommand{\ask}[1]{\mathrm{ask}\left( #1 \right)}
\newcommand{\bidGamma}[1]{\mathrm{bid}_\gamma\left( #1 \right)}
\newcommand{\askGamma}[1]{\mathrm{ask}_\gamma\left( #1 \right)}
\newcommand{\call}[1]{\mathcal{C}\left( #1 \right)}
\newcommand{\putt}[1]{\mathcal{P}\left( #1 \right)}
\newcommand{\PV}[1]{\mathrm{PV}\left( #1 \right)}
\newcommand{\bidcall}{b_\mathcal{C}}
\newcommand{\bidput}{b_\mathcal{P}}
\newcommand{\askcall}{a_\mathcal{C}}
\newcommand{\askput}{a_\mathcal{P}}
\newcommand{\sigmaput}{\sigma_{\mathcal{P}}}
\newcommand{\sigmacall}{\sigma_{\mathcal{C}}}

\newcommand{\callBSarg}[1]{\mathcal{C}^{\text{BS}}\left( #1 \right)}
\newcommand{\puttBSarg}[1]{\mathcal{P}^{\text{BS}}\left( #1 \right)}
\newcommand{\callBarg}[1]{\mathcal{C}^{\text{B}}\left( #1 \right)}
\newcommand{\puttBarg}[1]{\mathcal{P}^{\text{B}}\left( #1 \right)}
\newcommand{\callBS}{\mathcal{C}^{\text{BS}}}
\newcommand{\puttBS}{\mathcal{P}^{\text{BS}}}
\newcommand{\callB}{\mathcal{C}^{\text{B}}}
\newcommand{\puttB}{\mathcal{P}^{\text{B}}}

\begin{document}
\date{30 April 2021}
\title{Liquidity-free implied volatilities: an approach using conic finance}

\author[,1,2]{Matteo Michielon\thanks{matteo.michielon@nl.abnamro.com (corresponding author).}}
\author[,2]{Asma Khedher\thanks{a.khedher@uva.nl.}}
\author[,2,3]{Peter Spreij\thanks{p.j.c.spreij@uva.nl.}}
\affil[1]{\small Quantitative Analysis and Quantitative Development, ABN AMRO Bank N.V., Gustav Mahlerlaan 10, 1082 PP Amsterdam, The Netherlands.}
\affil[2]{\small Korteweg-de Vries Institute for Mathematics, University of Amsterdam, Science Park 105-107, 1098 XG Amsterdam, The Netherlands}
\affil[3]{\small Institute for Mathematics, Astrophysics and Particle Physics, Radboud University Nijmegen, Huygens building, Heyendaalseweg 135, 6525 AJ Nijmegen, The Netherlands}

\maketitle


\begin{abstract}
We consider the problem of calculating risk-neutral implied volatilities of European options without relying on option mid prices but solely on bid and ask prices. We provide an approach, based on the conic finance paradigm, that allows to uniquely strip risk-neutral implied volatilities from bid and ask quotes, and that does not require restrictive assumptions. Our methodology also allows to jointly calculate the implied liquidity of the market. The idea outlined in this paper can be applied to calculate other implied parameters from bid and ask security prices as soon as their theoretical risk-neutral counterparts are strictly increasing with respect to the former.
\end{abstract}


\begin{doublespace}

\section{Introduction}\label{sec:introduction}

The implied volatility of an option is defined as the value of the volatility of the underlying asset which, when used as input in a given pricing model, returns a theoretical value matching the current market price of the option considered. In this article we propose a methodology which allows to compute risk-neutral implied volatilities of European-options.\footnote{Here options are always assumed to have European-style exercise. Thus, the exercise type will be often omitted, for brevity.} This is accomplished without relying on any mid quote approximations. Instead, our approach can be applied starting from bid and ask quotes directly, and we outline how to use our technique under both Black-Scholes and Bachelier modeling settings.

The concept of implied volatility is relevant for different reasons. First of all, Black-Scholes (but also Bachelier) implied volatilities are important \emph{quoting conventions} in financial markets. They are therefore useful as benchmarks for the calibration of option pricing models. Nonetheless, several applications of the notion of implied volatility have been investigated outside the valuation framework, and, more precisely, in forecasting analysis: there are studies investigating the use of implied volatilities to predict, amongst other things, realized volatilities \cite{evidence}, asset returns \cite{an, fu} and financial market bubbles \cite{bubbles}. Moreover, implied volatility spreads, i.e., the differences in call and put implied volatilities, have been used to forecast option returns \cite{doran} and equity premia \cite{cao}. Thus, the more accurately one can calculate implied volatilities from option prices, especially far from the at-the-money point where liquidity is lower, the better. 

In practical applications and analyses, implied volatilities are often calculated starting from mid option prices, as for instance in \cite{midVols}. It is a known fact that the risk-neutral price of a contract lies within an interval with lower and upper bounds given by the bid and the ask prices, respectively. However, the risk-neutral price, in general, does not coincide with the mid, despite the latter is, usually, employed as a proxy for the former. To be able to calculate the correct risk-neutral implied volatility of an option without relying on mid market approximations, one would need to model option prices in a two-price economy. One possible approach to do so is that of conic finance, introduced in \cite{cherny}. This allows to evaluate bid and ask prices of contingent claims by recognizing that, in an economy, risk cannot be fully eliminated. Therefore, markets should quote based on the notions of (static) \emph{index of acceptability} and \emph{coherent risk measure}, consistently with the risk-neutral paradigm. By characterizing the structure of the contingent claims that are considered acceptable by the market, computing bid and ask prices can be performed by means of Choquet expectations \cite{choquet1953} of the relevant terminal payoffs with respect to distorted versions of the risk-neutral distribution of the underlying asset. This static approach to conic finance has found disparate practical applications. These applications range from exotic and structured products \cite{exotic1,structured} to contingent convertibles \cite{conicCoconuts}, from capital calculations \cite{capital1,capital2} to credit valuation adjustments \cite{cva1,cva2}, and again from hedging insurance risk \cite{insurance} to implied liquidity \cite{impliedLiquidity}. For a better overview of the applications of conic finance, see \cite{bookConicFinance}. We observe that the approach to conic finance based on static indices of acceptability can be extended to a time-dependent framework, as in \cite{dynamicConicFinance}, via the notion of \emph{dynamic index of acceptability} \cite{dynamicAccIndex}. However, as our aim is that of extracting market information from the currently-available market data (i.e., European option prices), a static approach to conic finance already suits our needs.

Our approach to imply risk-neutral volatilities without relying on any mid quote approximation is based on the conic finance theory of \cite{cherny} and, in particular, on \cite{michielon}, where a methodology to imply risk-neutral default probability distributions from bid and ask credit default swaps (CDSs) is outlined. Note that our methodology can be used to compute risk-neutral market-implied quantities from quoted bid and ask prices of any type of contingent claim, provided that some basic assumptions are satisfied. In particular, in the specific case of European options, our methodology requires some technical conditions to be fulfilled concerning the liquidity level of the market and the infima and suprema of the option prices with respect to changes in the volatility parameter. We observe that the fact that European option prices are strictly increasing with respect to the volatility parameter is essential for our technique to be applied (therefore, for other products, to imply a given model parameter a monotonicity condition needs to be satisfied; see Theorem \ref{th:general} for the technical details).

Within the conic finance framework the concept of conic implied volatility has been introduced in \cite[Sec. 5.4.3]{bookConicFinance}. Therein it is illustrated how, given market bid and ask quotes, a distortion function, and a liquidity level, one can compute the implied volatilities that allow to price back the observed bid and ask prices, named \emph{conic implied volatilities}. This approach is different from that of calculating \emph{bid} and \emph{ask implied volatilities}, that is, the implied volatilities that allow to match quoted bid and ask option prices under risk-neutral settings. Bid implied volatilities are lower than ask implied volatilities given that bid prices are below their ask counterparts. On the contrary, \cite[Sec. 5.4.3]{bookConicFinance} show that this condition does not need to hold when conic implied volatilities are calculated. Note that the technique proposed in \cite[Sec. 5.4.3]{bookConicFinance} is outlined in the specific case of options with European exercise features under Black-Scholes specifications for underlyings paying continuous dividends. However, it can be also applied, for instance, in the case of Bachelier specifications, and it is not restricted to the equity asset class only. Our idea is fundamentally different from both the above, as in our approach we imply a single volatility given a bid and an ask, and not an implied volatility per quote (i.e., one for the bid and one for the ask) as done in the standard case and in \cite[Sec. 5.4.3]{bookConicFinance}. This because we are interested in computing implied volatilities which can be interpreted as risk-neutral ones. Further, our approach still allows to compute implied volatility spreads, as the methodology can be followed for calls and puts separately. Our method guarantees that implied risk-neutral volatilities and liquidity levels, the latter in the spirit of \cite{impliedLiquidity}, can be uniquely determined. Further, the methodology outlined here is also simple from a computational perspective, as it only requires to solve a (constrained) non-linear system with two equations and two unknowns. 

We highlight here that it is not our intention to advocate the usage of Black-Sholes (or Bachelier) settings in financial modeling. The reason why we provide a method to strip \emph{liquidity-free} implied risk-neutral volatilities is that option prices are often quoted in terms of Black-Scholes (or Bachelier) volatilities. In addition, both Black-Scholes and Bachelier settings can be seen as option price interpolators, and the implied volatilities they generate are often benchmark inputs in several pricing models. Therefore, their accurate calculation, which we show can be performed without relying on mid quote approximations, is of key importance in financial modeling.

\medskip

This paper is organized as follows. Section \ref{sec:bidAsk} provides a brief introduction to the theory of conic finance. Section \ref{sec:liquidityFreeVols} outlines how to compute risk-neutral implied volatilities starting from bid and ask option prices using the conic finance theory. In particular, it highlights how to do so in the case distortions are modeled as Wang transforms \cite{wang} by recalling the conic Black Scholes formulae of \cite[Sec. 5.4]{bookConicFinance} and by providing conic Bachelier option pricing formulae. Section \ref{sec:example} provides an illustration of the methodology outlined in this article, while Section \ref{sec:conclusion} concludes. The proof of Theorem \ref{th:general} can be found in Appendix \ref{sec:proof}. For completeness, in Appendix \ref{sec:derivation} the derivation of risk-neutral Bachelier option pricing formulae is available, while in Appendix \ref{sec:remark} a remark on a property of the Wang transform is provided.

\section{Pricing in a two-price economy}\label{sec:bidAsk}
The theory of conic finance introduced in \cite{cherny2009} is based on the idea that, in financial markets, risks cannot be fully hedged. Therefore, positions are taken after having weighted the possible risks and rewards connected to the instruments traded in the market. Hence, financial markets are modeled as abstract counterparties that allow trades to take place after they have passed some sort of ``quality assessment''. To do so, financial markets would need a machinery to perform this appraisal. In \cite{cherny2009} this is based on the concept of index of acceptability. Given a probability space $(\Omega, \mathcal{F}, \mathbb{P})$, a functional $\alpha: L^\infty(\Omega, \mathcal{F}, \mathbb{P}) \to [0,+\infty]$ which assigns higher (lower) values to random variables that are expected to perform better (worse) is what an index of acceptability is. \cite{cherny2009} impose some technical conditions on the notion of index of acceptability, otherwise their class would be too wide for being of practical use. In particular, if two random cashflows are acceptable at a given level $\gamma$ (i.e., if their level of acceptability is at least $\gamma$), then the same applies to any convex combination of them. Moreover, indices of acceptability are assumed to be monotonic: if a random cashflow always outperforms a second one, then the former would be better ranked than the latter. Indices of acceptability are also supposed to be scale-invariant, i.e., the expected performance of a cashflow $X$ is the same of that of $\lambda X$, for every $\lambda>0$. Finally, the technical Fatou property needs to be satisfied: given $(X_n)_n$ a sequence of random cashflows such that, for every $n$, $|X_n|\leq 1$ and $\alpha(X_n)\geq \gamma$, then if $(X_n)_n$ converges in probability to a random cashflow $X$, then also $\alpha(X)\geq \gamma$. \cite{cherny2009} prove that, provided an index of acceptability $\alpha$, for every $x\geq0$ there exists a set $\mathfrak{Q}_x$ of probability measures absolutely continuous with respect to $\mathbb{P}$ such that $x\leq x'$ implies $\mathfrak{Q}_x\subseteq\mathfrak{Q}_{x'}$, and that
\begin{equation}\nonumber
\alpha(X)=\sup\left\{x\geq0:\inf_{\mathbb{Q}\in\mathfrak{Q}_x}\mathbb{E}^{\mathbb{Q}}(X)\geq0\right\}.
\end{equation}

The concept of index of acceptability can be then linked to that of coherent risk measure, i.e., a map $\rho:L^\infty(\Omega, \mathcal{F}, \mathbb{P}) \to [0,+\infty]$ that is transitive, sub-additive, positively homogeneous and monotonic, see \cite[Sec. 4.1]{bookConicFinance}. \cite{delbaen2009} show that a coherent risk measure can be identified with a functional such that $X\mapsto\sup_{\mathbb{Q}\in\mathfrak{Q}}\mathbb{E}^{\mathbb{Q}}(X)$, where the set $\mathfrak{Q}$ contains measures that are absolutely continuous with respect to $\mathbb{P}$. The concepts of index of acceptability and that of coherent risk measure can be tied together via the relationship
\begin{equation}\label{eq:ia}
\alpha(X)=\sup\left\{ x\geq0:\rho_x(-X)\leq0 \right\},
\end{equation} 
where $(\rho_x)_{x\geq0}$ is a family of coherent risk measures such that $\rho_x(-X)\leq\rho_{x'}(-X)$ whenever $x\leq x'$. From this it follows that $\alpha(X)\geq\gamma$ is equivalent to $\rho_\gamma(-X)\leq0$. For this reason, indices of acceptability with the aforementioned properties are often called \emph{coherent indices of acceptability}.

We now recall the definition of (asymmetric) Choquet integral \cite{choquet1953} (we will, from here onwards, always omit the word ``asymmetric'', as symmetric Choquet integrals are not relevant in this framework, and we refer the interested reader to \cite[Sec. 7]{denneberg}). For a \emph{non-additive probability} $\mu$ and a random variable $X$, the Choquet integral is defined as
\begin{equation}\label{eq:choquetGeneral}
	(\text{C}) \int_\Omega X \, d\mu \coloneqq \int_{-\infty}^0 \mu(X\geq t)-1\,dt + \int_0^{+\infty} \mu(X\geq t)\,dt.
\end{equation}
In \eqref{eq:choquetGeneral}, the integrals on the right-hand side should be interpreted as improper Riemann integrals. Therefore, they both exist given that their arguments are monotonic functions, which guarantees that the sets of their discontinuities have a Lebesgue measure of zero. Note, however, that their sum does not necessarily exist\footnote{If $X$ is non-negative(positive), then \eqref{eq:choquetGeneral} is guaranteed to be well-defined.}: see \cite[Sec. 5]{denneberg} for a detailed treatment of Choquet integrals.

Let a risk-neutral measure $\mathbb{Q}\in\bigcap_{x\geq0}\mathfrak{Q}_x$. A \emph{distortion function} is a function from $[0,1]$ to $[0,1]$ that maps 0 to 0 and 1 to 1. For a concave distortion function $\psi(\,\cdot\,)$ we denote with $\psi(\mathbb{Q})(A)$ the (potentially non-additive) probability measure that assigns to each measurable set $A$ the probability mass $\psi(\mathbb{Q}(A))$. Given $\left(\psi_x\right)_{x\geq 0}$ an increasing family of concave distortion functions, the map $\rho_x$ such that $X\mapsto (\text{C}) \int_\Omega X \, d\psi_x({\mathbb{Q}})$ defines a coherent risk measure. Hence, as per \cite{cherny2009}, functionals of this form can be employed to describe indices of acceptability by setting
\begin{equation}\label{eq:oia}
\alpha(X)\coloneqq\sup\left\{x\geq0:	(\text{C})\int_\Omega -X \, d\psi_x({\mathbb{Q}})\leq0\right\}.
\end{equation}
The tools just introduced can be now used to characterize direction-dependent pricing in financial markets.

Assume that a threshold of at least $\gamma$ has been set by the market for a given contingent claim to be considered acceptable and, thus, tradable. We assume a constant risk-free rate $r$\footnote{Note that the constant risk-free rate assumption has been made only for consistency with the fact that, in this article, we consider Black-Scholes and Bachelier models. However, in the case of a time-dependent risk-free rate, all the steps outlined from here onwards would still hold.} and consider a contingent claim $X$ with a terminal payoff at time $T$. The market is then willing to buy $X$ at a price $b$ if and only if $\alpha(X-e^{-rT} b)\geq\gamma$. This is equivalent, given the assumption that the market evaluates the performance of contingent claims by means of Choquet integrals, to  the condition $
b\leq-e^{-rT}(\mathrm{C})\int_\Omega -X\,d\psi_\gamma(\mathbb{Q})$. Thus, the bid price of $X$, $\mathrm{bid}_\gamma(X)$, equals
\begin{equation}\label{eq:bid}
\mathrm{bid}_\gamma(X)=-e^{-rT}(\mathrm{C})\int_\Omega -X\,d\psi_\gamma(\mathbb{Q}).
\end{equation}
As the ask price of $X$, $\mathrm{ask}_\gamma(X)$, equals $-\mathrm{bid}_\gamma(-X)$, from \eqref{eq:bid} it then immediately follows that
\begin{equation}\label{eq:ask}
\askGamma{X}=e^{-rT}(\mathrm{C})\int_\Omega X\,d\psi_\gamma(\mathbb{Q}).
\end{equation}
Note that should more than one risk-neutral measure exist, then one would need to choose which risk-neutral measure to use within formulae \eqref{eq:bid} and \eqref{eq:ask}. Further, we observe that, in the formulae provided above, the choice of the distortion function provides the modeler with a degree of freedom to describe the liquidity dynamics of the market. In addition, to different values of the distortion parameter $\gamma$ there correspond different market liquidity specifications. This framework reminds that of modeling preferences towards risk by means of utility functions. Despite utility theory characterizes agents' behavior from a micro-economic perspective (i.e., the individual preferences of each agent) while conic finance describes risk attitudes of financial markets, there are some similarities between the two approaches worth of attention. In particular, in utility theory the modeler has to choose the functional form of the utility function to be used. This is similar to the conic finance case, where a choice related to the distortion function also has to be made. Further, in utility theory one has to choose the parameter(s) of the utility function in order to describe the level of risk aversion (or risk tolerance) of an agent. Similarly, on the conic finance side, the behavior of the market is further described by the distortion parameter $\gamma$. In addition, we also point out that Choquet integrals are common tools in decision theory. In particular, we recall the results of \cite{schmeidler89} which characterize choices under uncertainty, the latter in the sense of  \cite{knight}, in terms of Choquet integrals (for a representation result concerning Choquet integrals, see \cite{schmeidler86}, on which \cite{schmeidler89} is based on). We further highlight that Choquet integrals can be also applied to option pricing problems under uncertainty by means of Choquet Brownian motions, introduced in \cite{choquetBrownian}, as done in \cite{choquetOptions}.

\section{Liquidity-free option implied volatilities}\label{sec:liquidityFreeVols}

From here onwards we consider European options only, and we assume to be either within the Black-Scholes or the Bachelier framework. This because we are interested in backing out either log-normal or normal implied volatilities. 

Given a filtered probability space $(\Omega, \mathcal{F}, (\mathcal{F}_t)_{t\in[0,T]}, \mathbb{P})$, we denote with $(X_t)_{t\in[0,T]}$ the process representing a ``generic'' underlying. In particular, by introducing an adjusted risk-neutral drift $r-\alpha$, one can define Black-Scholes dynamics via the stochastic differential equation (SDE) given, under the risk-neutral measure $\mathbb{Q}$ equivalent to $\mathbb{P}$ on $\mathcal{F}$, by
\begin{equation}\label{eq:bs}
dX = (r-\alpha) X dt + \sigma X dW.
\end{equation}
In \eqref{eq:bs}, $\sigma$ denotes the volatility term, and $(W_t)_{t\in[0,T]}$ a Brownian motion adapted to the filtration $(\mathcal{F}_t)_{t\in[0,T]}$. The parameter $\alpha$ can be defined according to the asset class considered. For instance, setting $\alpha=0$ outlines the standard Black-Scholes framework on a non-dividend paying underlying, setting $\alpha=q$ with $q$ denoting the continuous dividend yield corresponds to the Black-Scholes framework for an underlying paying continuous dividends, while setting $\alpha=r$ corresponds to the Black model for futures options; see \cite[Sec. 1.1.6]{optionFormulae} for further possible specifications.

To take into accunt the possibility that the underlying asset can reach negative values, e.g., in the case of rates and oil prices\footnote{See \url{https://www.cmegroup.com/content/dam/cmegroup/notices/clearing/2020/04/Chadv20-152.pdf} for a note of the Chicago Mercantile Exchange concerning the possible use of the Bachelier formula due to the negative oil prices observed in 2020.}, option prices can also be quoted in terms of Bachelier (i.e., normal) implied volatilities. Therefore, in a similar manner as per \eqref{eq:bs} and using the same notation conventions, one can define the Bachelier SDE as
\begin{equation}\label{eq:b}
dX = (r-\alpha) X dt + \sigma dW.
\end{equation}
Via \eqref{eq:b} we have chosen to describe a generalized variant, in the sense of \cite[Sec. 1.1.6]{optionFormulae}, of the ``contemporary'' version of the Bachelier model as in \cite[Sec. 3.3]{musiela}. Note, however, that in the literature sometimes the SDE corresponding to the Bachelier model slightly differs from that outlined in \eqref{eq:b}, for instance by not considering the drift term (see \cite[Sec. 1.3.1]{optionFormulae}). In any case, independently on the exact specifications of the Bachelier SDE considered, all the Bachelier-related calculations available in this article can be performed in the same manner, up to minor rearrangements.

For a given strike $K$ and maturity $T$, we denote with $\mathcal{C}$ and with $\mathcal{P}$ the prices of a call and a put option written on $X$ with such strike and maturity. However, when considering the Black-Scholes model \eqref{eq:bs} (Bachelier model \eqref{eq:b}), we will use $\callBS$ ($\callB$) and $\puttBS$ ($\puttB$), instead. Note that, depending, on the context, we will make option prices explicitly depend on specific parameters only, as it will be clear in the next sections. This to keep the notation as light as possible. In any case, the dependency on both strike and maturity will be always omitted, as redundant in our context.

%

In practice, implied volatilities, either normal or log-normal, are backed out from call and put options separately. More precisely, implied volatilities for call options are computed starting from the mid prices of these options and, similarly, the same applies to put implied volatilites. However, in principle, one would like to compute the real risk-neutral implied volatilities, without relying on approximating risk-neutral prices by their mid counterparts. A similar problem to this has been analyzed in \cite{michielon}. Therein it is shown that, in the case of CDSs, under mild assumption concerning the liquidity level of the market and the characteristics of the default time process, it is possible to strip risk-neutral default probabilities from bid and ask CDS quotes directly in a unique manner. The considerations available in \cite{michielon} are now extended to a more general setup. In particular, the methodology we highlight here is quite general. That  is, it can be applied to any contingent claim whose risk-neutral price depends on a single unknown parameter provided that the former is strictly increasing with respect to the latter, and as soon as two basic additional conditions are satisfied. That is, the range of theoretical prices obtainable by changing the free parameter, as well as the liquidity level of the market, should be ``wide enough'', as we will explain more technically in Theorem~\ref{th:general}.

Let $Y$ be a contingent claim, and denote with $\PV{Y(\lambda)}$ its risk-neutral price, assumed dependent on an unknown parameter $\lambda$. Further, let $b$ and $a$ denote its quoted bid and ask prices, respectively. The main result available in \cite{michielon} is recalled in Theorem \ref{th:general} in a more general fashion. The proof of Theorem \ref{th:general} follows from Lemmas 1, 2, 3 and Theorem 1 in \cite{michielon}, as the steps outlined therein can be followed in the same manner. For completeness, we have provided the aforementioned results, adapted to the more general context considered in the present article, in Appendix \ref{sec:proof}.

\begin{theorem}\label{th:general}
Let $Y$ be a contingent claim whose price depends on a parameter $\lambda>0$ such that the risk-neutral price of $Y$ is strictly increasing with respect to $\lambda$. Assume that $\inf_{\lambda>0}\PV{Y(\lambda)}<b$ and that $\sup_{\lambda>0}\PV{Y(\lambda)}>a$. This uniquely identifies an interval $[\lambda_a, \lambda_b]$ such that $\lambda\in[\lambda_b, \lambda_a]$ if and only if $\inf_{\lambda>0}\PV{Y(\lambda)}<b$ and $\sup_{\lambda>0}\PV{Y(\lambda)}>a$. Moreover, assume that for every $\lambda\in[\lambda_b, \lambda_a]$ there exists $\gamma>0$ such that $\ask{Y(\lambda),\gamma}-\bid{Y(\lambda),\gamma}=a-b$. Then, the constrained non-linear system 
\begin{equation}\label{eq:system}
\begin{cases}
\bid{Y(\lambda),\gamma} =b \\
\ask{Y(\lambda),\gamma} =a            
\end{cases}
\end{equation}
with 
\begin{equation}\nonumber
b<\PV{Y(\lambda)}<a
\end{equation}
admits a solution, which is also unique.
\end{theorem}

We observe that both call and put option prices are strictly increasing with respect to the volatility of the underlying.\footnote{Note that the conditions $\inf_{\lambda>0}\PV{Y(\lambda)}<b$ and that $\sup_{\lambda>0}\PV{Y(\lambda)}>a$ can be made more explicit in the case of European options. In particular, for a call option it results that $\inf_{\sigma>0}\call{\sigma}= e^{-rT}(e^{(r-\alpha) T}X_0-K)^+$, while $\sup_{\sigma>0}\call{\sigma}=e^{-\alpha T}X_0$. Similarly, for a put option it results that $\inf_{\sigma>0}\putt{\sigma}= e^{-rT}(K-e^{(r-\alpha) T}X_0)^+$, and that $\sup_{\sigma>0}\putt{\sigma}=e^{-rT}K$.} Denote with $b_\mathcal{C}$ ($b_\mathcal{P}$) and with $a_\mathcal{C}$ ($a_\mathcal{P}$) the quoted bid and ask price of the call (put), respectively. Ideally, one would aim to find an implied risk-neutral volatility $\sigma$ and two distortion parameters $\gamma_\mathcal{C}$ and $\gamma_\mathcal{P}$ such that the equalities
\begin{equation}\label{eq:system4equations}
\begin{cases}
\bid{\call{\sigma},\gamma_\mathcal{C}} =\bidcall \\
\ask{\call{\sigma},\gamma_\mathcal{C}} =\askcall  \\
\bid{\putt{\sigma},\gamma_\mathcal{P}} =\bidput \\
\ask{\putt{\sigma},\gamma_\mathcal{P}} =\askput     
\end{cases}
\end{equation}
with the constraints
\begin{equation}\nonumber
\begin{cases}
\bidcall<\PV{\call{\sigma}}<\askcall\\
\bidput<\PV{\putt{\sigma}}<\askput
\end{cases}
\end{equation}
are satisfied. Notwithstanding, it is in general not possible to solve \eqref{eq:system4equations} due to the obvious lack of degrees of freedom. Therefore, separate volatility and liquidity parameters should be used for calls and puts, as done in practice to compute call-put volatility spreads. In particular, one can solve
\begin{equation}\nonumber
\begin{cases}
\bid{\call{\sigma},\gamma} =b_\mathcal{C} \\
\ask{\call{\sigma},\gamma} =a_\mathcal{C}  \\
\end{cases}
\end{equation}
with 
\begin{equation}\nonumber
b_\mathcal{C}<\PV{\mathcal{C}(\sigma)}<a_\mathcal{C},
\end{equation}
and obtain that a unique solution, in virtue of Theorem \ref{th:general}, exists, denoted with $(\sigma_\mathcal{C}, \gamma_\mathcal{C})$. Similarly, one can solve
\begin{equation}\nonumber
\begin{cases}
\bid{\putt{\sigma},\gamma} =b_\mathcal{P} \\
\ask{\putt{\sigma},\gamma} =a_\mathcal{P}     
\end{cases}
\end{equation}
with 
\begin{equation}\nonumber
b_\mathcal{P}<\PV{\mathcal{P}(\sigma)}<a_\mathcal{P}
\end{equation}
separately, and again obtain a unique solution, due to Theorem \ref{th:general}, denoted as  $(\sigma_\mathcal{P}, \gamma_\mathcal{P})$.  We call $\sigmacall$ and $\sigmaput$ the liquidity-free call and put implied volatilities, respectively. The quantities $\gamma_\mathcal{C}$ and $\gamma_\mathcal{P}$ denote the implied liquidity levels of the market for calls and puts, respectively.

The approach proposed here is different from that proposed in \cite[Sec. 5.4.3]{bookConicFinance}. Therein, the notion of conic (Black-Scholes) implied volatility is introduced. In particular, for a fixed (and known) distortion parameter, one can then imply a volatility for the bid and one for the ask. However, our approach allows to simultaneously imply both the distortion parameter and the implied volatility directly, without therefore relying on an initial estimation procedure for the distortion itself. This because our goal is that of computing implied volatilities that can be interpreted as risk-neutral ones.

In time series analysis call-put volatility spreads, as outlined in Section \ref{sec:introduction}, can be used as regression variables for forecasting analysis. By taking into account the approach outlined here one would not only have the possibility to introduce liquidity-free call-put volatility spreads in the regression model considered, but also to take into account the implied distortion (i.e., liquidity) parameters as regression variables as well. Potentially, this could enhance the explanatory power of the regression models used for prediction purposes (see \cite{cva2} for an illustration of the explanatory power of the distortion parameter as far as liquidity is concerned).

\subsection{Implied volatilities with the Wang transform}\label{sec:formulae}
The choice of the distortion function to be used in the bid-ask calibration problem is arbitrary, provided that it is concave. Consequently, different possibilities are available, see \cite{cherny2009} and \cite[Sec. 4.7]{bookConicFinance}. However, for distributions of normal or log-normal random variables, which are often employed in financial applications, the Wang transform \cite{wang}, which is defined as 
\begin{equation}\label{eq:wang}
\psi_\gamma(x)\coloneqq \Phi(\Phi^{-1}(x)+\gamma)
\end{equation}
with $\Phi(\,\cdot\,)$ denoting the cumulative distribution function of a standard normal random variable, is a convenient choice. This because \eqref{eq:wang} still allows to obtain closed-form solutions for call and put option prices, see \cite[Sec. 5.4]{bookConicFinance}.\footnote{See Appendix \ref{sec:remark} for a remark concerning how the Wang transform can be a useful tool as soon as the distribution of a normal random variable is transformed via a non-decreasing and left-continuous function.} Therefore, under both Black-Scholes \eqref{eq:bs} and Bachelier \eqref{eq:b} settings, exact formulae can be used to calculate bid and ask option prices via the Wang transform.\footnote{Some other cases where the Wang transform produces analytical option prices formulae are those of the Sprenkle, Boness and Samuelson models (see \cite[Sec. 1.31, 1.32 and 1.33]{optionFormulae}). However, note that computing the Wang transform is computationally expensive, as this requires the evaluation of both the cumulative distribution and quantile functions of a standard normal random variable. Therefore, for large datasets and when the Wang transform does not guarantee analytical formulae to exists, then other choices for the distortion function might be more convenient (see \cite[Sec. 4.7]{bookConicFinance} for an overview).} Thus, our procedure to back out implied volatilities (and, consequentially, implied distortion parameters) can be easily implemented, with the advantage that it does not require to compute the integrals \eqref{eq:bid} and \eqref{eq:ask} numerically should the Wang transform be used.

In the case of the Black-Scholes framework, one obtains that the risk-neutral price of a call option is given by 
\begin{equation}\nonumber
\callBSarg{\alpha}=e^{-\alpha T}X_0\Phi(d_+)-e^{-rT}K\Phi(d_-),
\end{equation}
where
\begin{equation}\nonumber
d_+\coloneqq\frac{\ln\left(\frac{X_0}{K}\right)+(r-\alpha+\frac{1}{2}\sigma^2) T}{\sigma\sqrt{T}},
\end{equation}
and with
\begin{equation}\nonumber
d_-\coloneqq d_+-\sigma\sqrt{T}=\frac{\ln\left(\frac{X_0}{K}\right)+(r-\alpha-\frac{1}{2}\sigma^2) T}{\sigma\sqrt{T}}.
\end{equation}
Further,
\begin{equation}\nonumber
\puttBSarg{\alpha}=e^{-rT}K\Phi(-d_-) -e^{-\alpha T}X_0\Phi(-d_+).
\end{equation}
\cite[Sec. 5.4.2]{bookConicFinance} obtain that, by considering the Wang transform under Black-Scholes settings, bid and ask prices for European calls and puts can be computed as $\bidGamma{\callBSarg{\alpha}}=\callBSarg{\alpha + \frac{\gamma\sigma}{\sqrt{T}}}$, $\askGamma{\callBSarg{\alpha}}=\callBSarg{\alpha - \frac{\gamma\sigma}{\sqrt{T}}}$, $\bidGamma{\puttBSarg{\alpha}}=\puttBSarg{\alpha - \frac{\gamma\sigma}{\sqrt{T}}}$ and, finally, $\askGamma{\puttBSarg{\alpha}}=\puttBSarg{\alpha + \frac{\gamma\sigma}{\sqrt{T}}}$.

We now provide similar relationships in the case the Bachelier model \eqref{eq:b} is considered.\footnote{Risk-neutral call and put option pricing formulae for the Bachelier model are available in Appendix \ref{sec:derivation}, for completeness.} Let $\cdf{X_T}{\,\cdot\,}$ denote the time-$T$ risk-neutral distribution of the underlying asset. First of all we recall that for European vanilla options, if the underlying can reach negative values, in line with \cite[Sec. 5.5]{bookConicFinance} the following formulae can be used to calculate bid and ask European option prices:
\begin{equation}\label{eq:bidCall}
\bidGamma{\mathcal{C}}=e^{-rT}\int_{K}^{\infty} (x-K)\,d\distortion{\cdf{X_T}{x}},
\end{equation}
\begin{equation}\label{eq:askCall}
\askGamma{\mathcal{C}}=e^{-rT}\int_{K}^{\infty} (K-x)\,d\distortion{1-\cdf{X_T}{x}},
\end{equation}
\begin{equation}\label{eq:bidPut}
\bidGamma{\mathcal{P}}=e^{-rT}\int_{-\infty}^{K} (x-K)\,d\distortion{1-\cdf{X_T}{x}},
\end{equation}
and
\begin{equation}\label{eq:askPut}
\askGamma{\mathcal{P}}=e^{-rT}\int_{-\infty}^{K} (K-x)\,d\distortion{\cdf{X_T}{x}}.
\end{equation}
Observe that under both the Black-Scholes and Bachelier specifications \eqref{eq:bs} and \eqref{eq:b} continuous probability density functions for the terminal risk-neutral distribution of the underlying asset are available. Therefore, the relationships \eqref{eq:bidCall}, \eqref{eq:askCall}, \eqref{eq:bidPut} and \eqref{eq:askPut} can be interpreted as both Riemann-Stieltjes and Lebesgue-Stieltjes integrals.

Under the Bachelier dynamics \eqref{eq:b} the risk-neutral distribution of the underlying, at time $T$, is normal with mean $\bar{\mu}$ and variance $\bar{\sigma}^2$ as per \eqref{eq:mean} and \eqref{eq:variance} in Appendix \ref{sec:derivation}. If we consider a Wang transformation with distortion parameter $\gamma$ we obtain that, at time $T$, the underlying $X_T$ is still normally distributed with the same variance $\bar{\sigma}^2$, but this time with mean given by $\bar{\mu}_-\coloneqq\bar{\mu}-\gamma\bar{\sigma}$, see \cite{wang}.
Therefore, we can apply relationship \eqref{eq:bidCall} and obtain that
\begin{align}
\bidGamma{\callB}&=e^{-rT}\int_{K}^{\infty} (x-K)\,d\distortion{\cdf{X}{x}}\nonumber\\
&=e^{-rT}\int_{K}^{\infty} \frac{x-K}{\sigmabar\sqrt{2\pi}}e^{-\frac{1}{2}\left( \frac{x-\bar{\mu}_-}{\bar{\sigma}} \right)^2}\,dx\nonumber\\
&=e^{-rT}\int_{\frac{K-\bar{\mu}_-}{\sigmabar}}^{\infty} \frac{\bar{\mu}_-+\sigmabar x -K}{\sqrt{2\pi}}e^{-\frac{x^2}{2}}\,dx\nonumber\\
&=e^{-rT}\left[(\bar{\mu}_--K)\int_{\frac{K-\bar{\mu}_-}{\sigmabar}}^{\infty} \frac{1}{\sqrt{2\pi}}e^{-\frac{x^2}{2}}\,dx-\sigmabar\int_{\frac{K-\bar{\mu}_-}{\sigmabar}}^{\infty} \frac{-x}{\sqrt{2\pi}}e^{-\frac{x^2}{2}}\,dx\right]\nonumber\\
&=e^{-rT}\left[(\bar{\mu}_--K)\Phi\left(\frac{\bar{\mu}_--K}{\sigmabar}\right)+\sigmabar\phi\left(\frac{\bar{\mu}_--K}{\sigmabar}\right)\right].\nonumber
\end{align}
We can now calculate the call ask price via \eqref{eq:askCall}. First we observe, see \cite{wang}, that
\begin{equation}\label{eq:helper}
\distortion{1-\cdf{X_T}{x}}=\distortion{1-\Phi\left(\frac{x-\mubar}{\sigmabar}\right)}=\distortion{\Phi\left(\frac{\mubar-x}{\sigmabar}\right)}=\Phi\left(\frac{\mubar-x+\gamma\sigmabar}{\sigmabar}\right).
\end{equation}
By setting $\bar{\mu}_+\coloneqq\bar{\mu}+\gamma\bar{\sigma}$ we obtain that
\begin{align}
\askGamma{\callB}&=e^{-rT}\int_{K}^{\infty} (K-x)\,d\distortion{1-\cdf{X_T}{x}}\nonumber\\
&=e^{-rT}\int_{K}^{\infty}\frac{x-K}{\sigmabar\sqrt{2\pi}}e^{-\frac{1}{2}\left(\frac{x-\mubar_+}{\bar{\sigma}}\right)^2}\,dx\nonumber\\
&=e^{-rT}\left[(\mubar_+-K)\Phi\left(\frac{\mubar_+-K}{\sigmabar}\right)+\sigmabar\phi\left(\frac{\mubar_+-K}{\sigmabar}\right)\right].\nonumber
\end{align}

We now calculate the ask price of an European put option via \eqref{eq:askPut}. It results that
\begin{align}
\askGamma{\puttB}&=e^{-rT}\int_{-\infty}^{K} (K-x)\,d\distortion{\cdf{X_T}{x}}\nonumber\\
&=e^{-rT}\int_{-\infty}^K \frac{K-x}{\sigmabar\sqrt{2\pi}}e^{-\frac{1}{2}\left(\frac{x-\mubar_-}{\sigmabar}\right)^2}\,dx\nonumber\\
&=e^{-rT}\int_{-\infty}^\frac{K-\mubar_-}{\sigmabar} \frac{K-\mubar_--\sigmabar x}{\sigmabar\sqrt{2\pi}}e^{-\frac{1}{2}x^2}\,dx\nonumber\\
&=e^{-rT}\left[(K-\mubar_-)\int_{-\infty}^\frac{K-\mubar_-}{\sigmabar}\frac{1}{\sigmabar\sqrt{2\pi}}e^{-\frac{1}{2}x^2}\,dx+\sigmabar\int_{-\infty}^\frac{K-\mubar_-}{\sigmabar}\frac{-x}{\sigmabar\sqrt{2\pi}}e^{-\frac{1}{2}x^2}\,dx\right]\nonumber\\
&=e^{-rT}\left[(K-\mubar_-)\Phi\left(\frac{K-\mubar_-}{\sigmabar}\right)+\sigmabar\phi\left(\frac{K-\mubar_-}{\sigmabar}\right)\right].\nonumber
\end{align}

Recalling \eqref{eq:helper}, the bid price of the put can be calculated using \eqref{eq:bidPut}, from which it follows that
\begin{align}
\bidGamma{\puttB}&=e^{-rT}\int_{-\infty}^{K} (x-K)\,d\distortion{1-\cdf{X_T}{x}}\nonumber\\
&=e^{-rT}\int_{-\infty}^{K}\frac{K-x}{\sigmabar\sqrt{2\pi}}e^{-\frac{1}{2}\left(\frac{x-\mubar_+}{\sigma}\right)^2}\,dx\nonumber\\
&=e^{-rT}\left[(K-\mubar_+)\Phi\left(\frac{K-\mubar_+}{\sigmabar}\right)+\sigmabar\phi\left(\frac{K-\mubar_+}{\sigmabar}\right)\right].\nonumber
\end{align}

To summarize, see notation in Appendix \ref{sec:derivation}, one obtains that $\bidGamma{\callB}=\callBarg{\mubar_-}$, $\askGamma{\callB}=\callBarg{\mubar_+}$, $\bidGamma{\puttB}=\puttBarg{\mubar_+}$, while $\askGamma{\puttB}=\puttBarg{\mubar_-}$.

\section{Example}\label{sec:example}
Here we show how liquidity-free implied volatilities can be extracted from bid and ask prices. In particular, we consider European options on four different underlyings, i.e., European call options on the S\&P 500 index, European put options on the FTSE MIB index, European call options on UBS shares, and European put options on Deutsche Telekom shares. For each of the cases considered we compute, for a given maturity (not kept unchanged for all the underlyings), bid and ask prices, risk-neutral and mid prices, absolute liquidity spreads, relative liquidity spreads, implied risk-neutral and mid volatilities, as well as implied distortion parameters. All the aforementioned calculations have been performed for all the quoted options available for which both bid and ask prices could be retrieved.\footnote{In this section plots have been constructed with respect to moneyness, defined here as the ratio between a given strike price and the value of the underlying.} The Wang transform has been chosen as distortion in all the cases analyzed.

We start by considering European call options on the S\&P 500 index, for which a wide range of strikes is available. These options are very liquid, as illustrated by Figures \ref{fig:SPX_bid_ask_call} and \ref{fig:SPX_absolute_spread_call} (note that the relative bid-ask spreads for deep out-of-the-money options in Figure \ref{fig:SPX_relative_spread_call} are large due to those options having small market value). This results in risk-neutral and mid prices that are very close to each other, as shown in Figure \ref{fig:SPX_price_call}. Also the risk-neutral and mid implied volatility smiles, see Figure \ref{fig:SPX_implied_volatility_call}, are basically overlapping, as expected. The implied distortion parameters, illustrated in Figure \ref{fig:SPX_implied_distortion_call}, closely follow the trend of the relative bid-ask spreads of Figure \ref{fig:SPX_relative_spread_call}.

\begin{figure}[H]
\begin{center}
   \begin{subfigure}{.5\linewidth}
     \centering
     \includegraphics[scale=0.5]{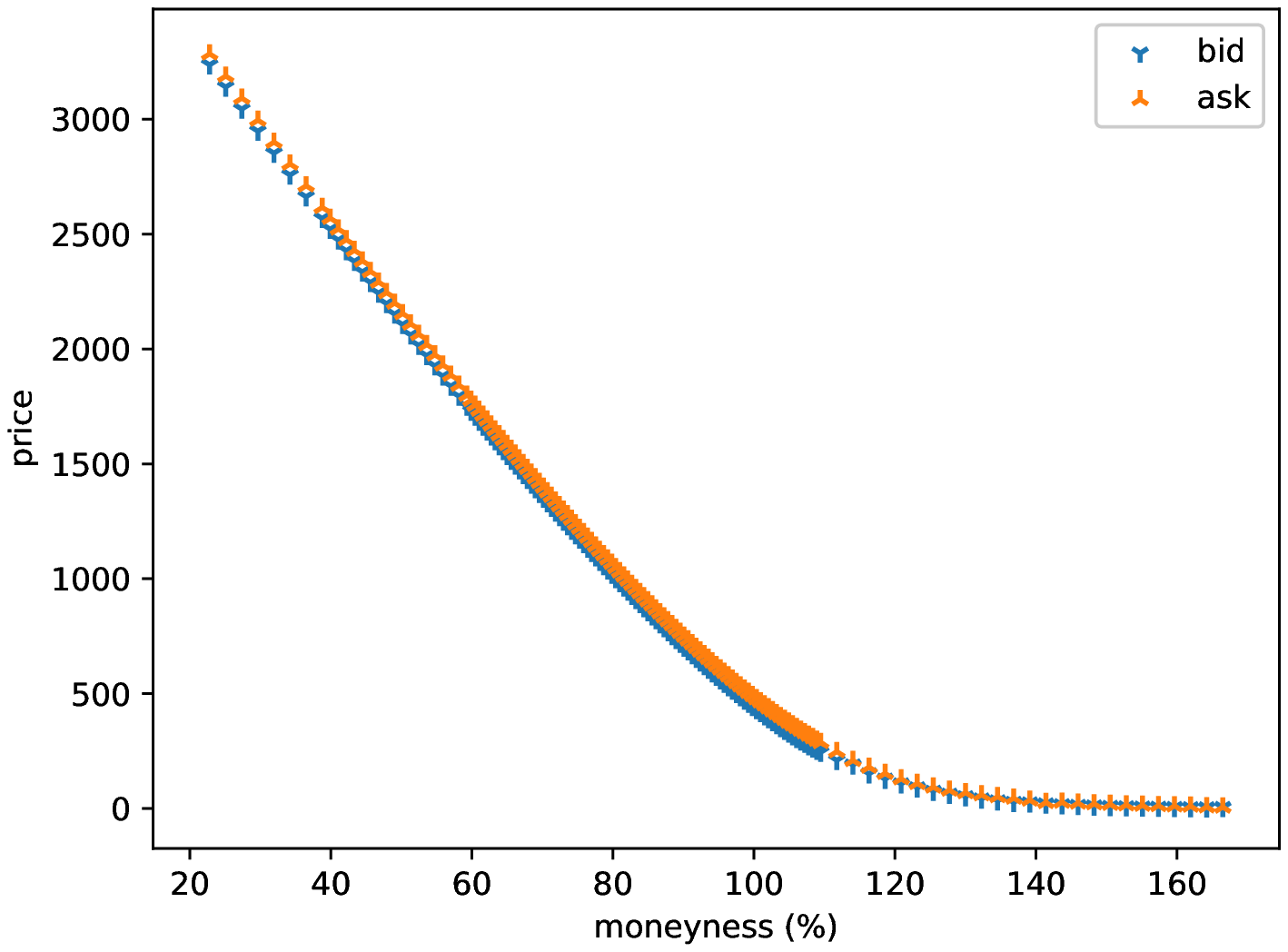}
     \caption{}\label{fig:SPX_bid_ask_call}
   \end{subfigure}
   \begin{subfigure}{.5\linewidth}
    \centering
     \includegraphics[scale=0.5]{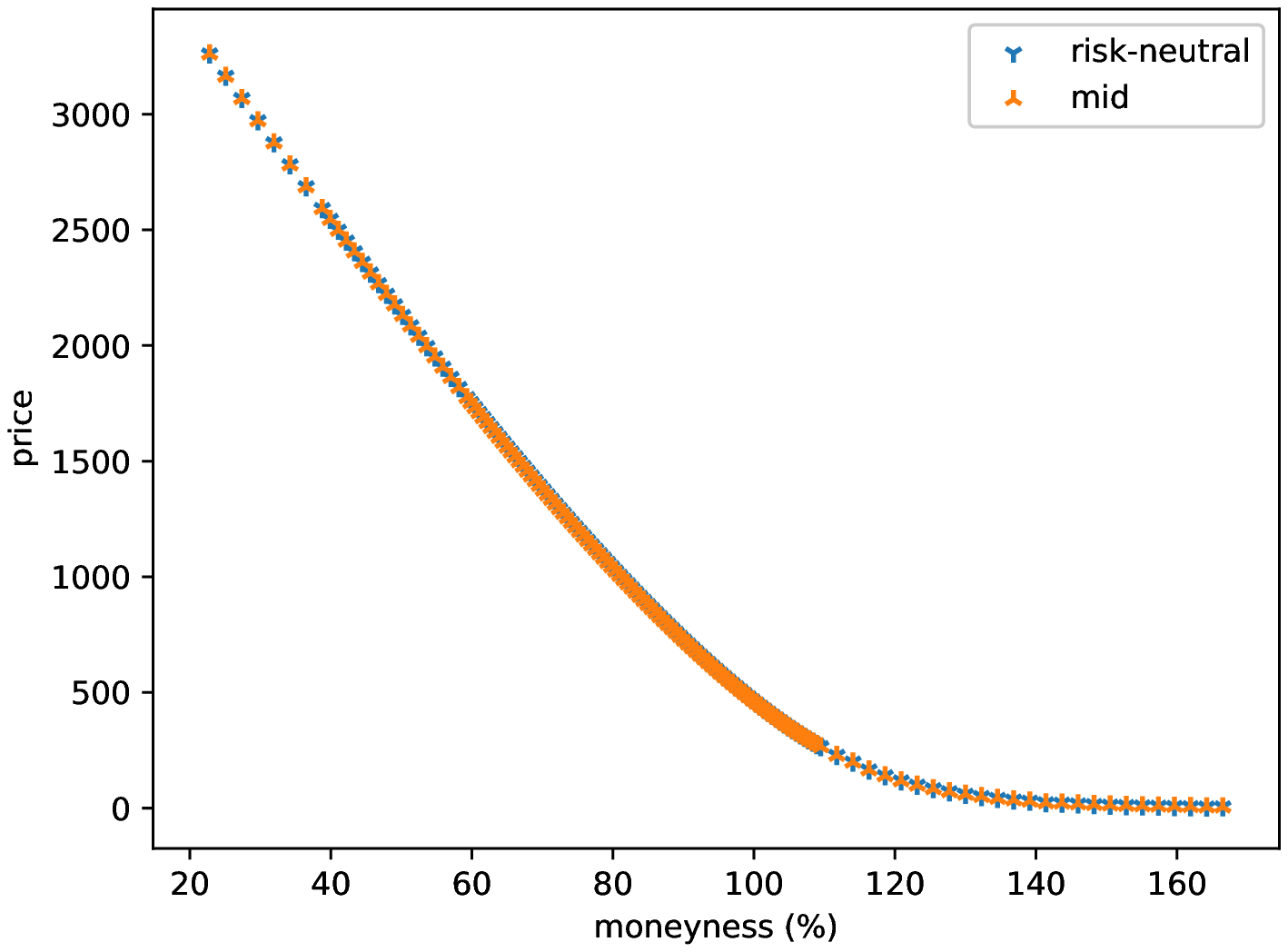}
     \caption{}\label{fig:SPX_price_call}
   \end{subfigure}
   \end{center}
      \caption{Bid and ask prices for European call options on the S\&P 500 index expiring in 886 days (Options Price Reporting Authority), panel (a), and their corresponding risk-neutral and mid counterparts, panel (b).\label{fig:SPX}}
\end{figure}

\begin{figure}[H]
\begin{center}
   \begin{subfigure}{.5\linewidth}
     \centering
     \includegraphics[scale=0.5]{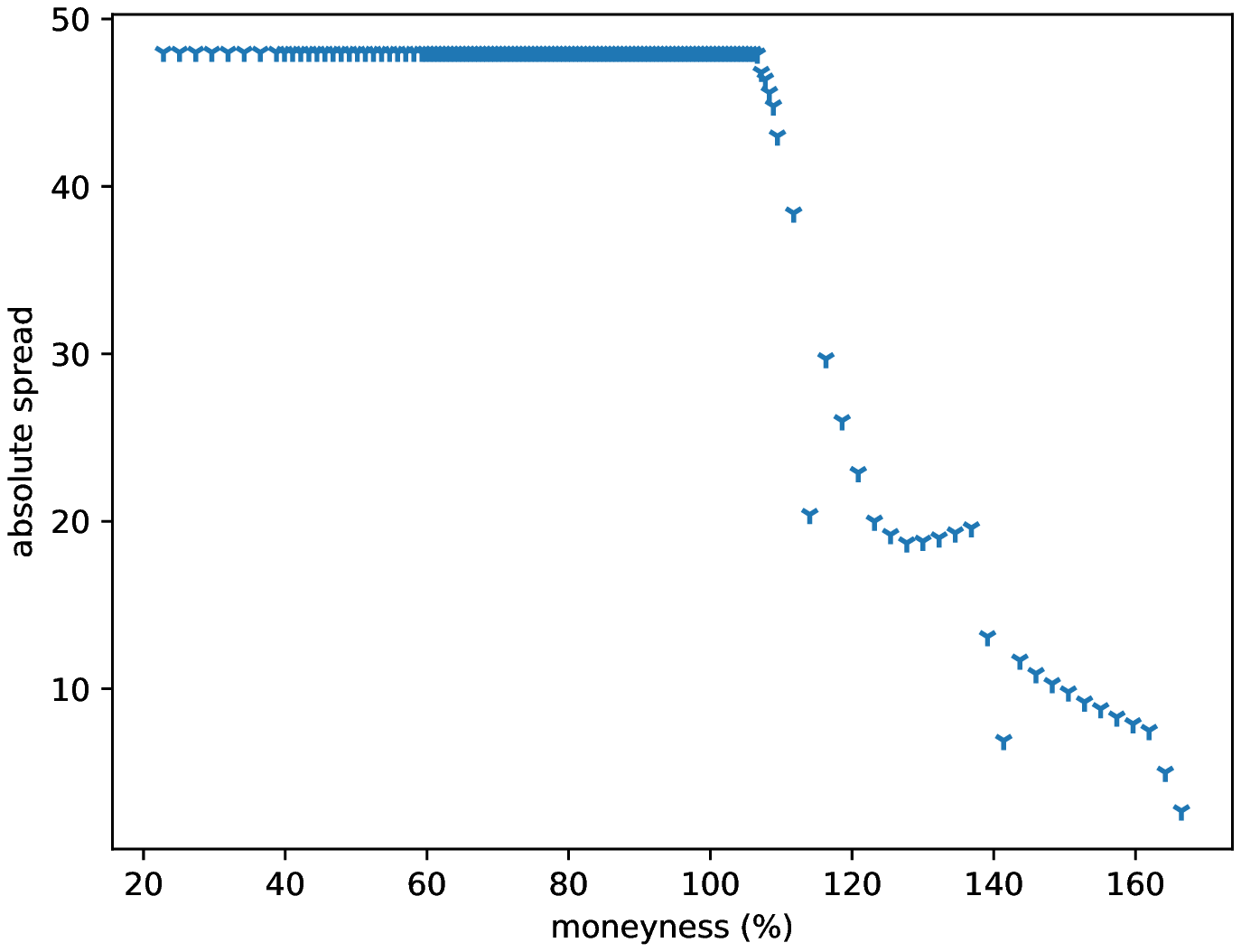}
     \caption{}\label{fig:SPX_absolute_spread_call}
   \end{subfigure}
   \begin{subfigure}{.5\linewidth}
    \centering
     \includegraphics[scale=0.5]{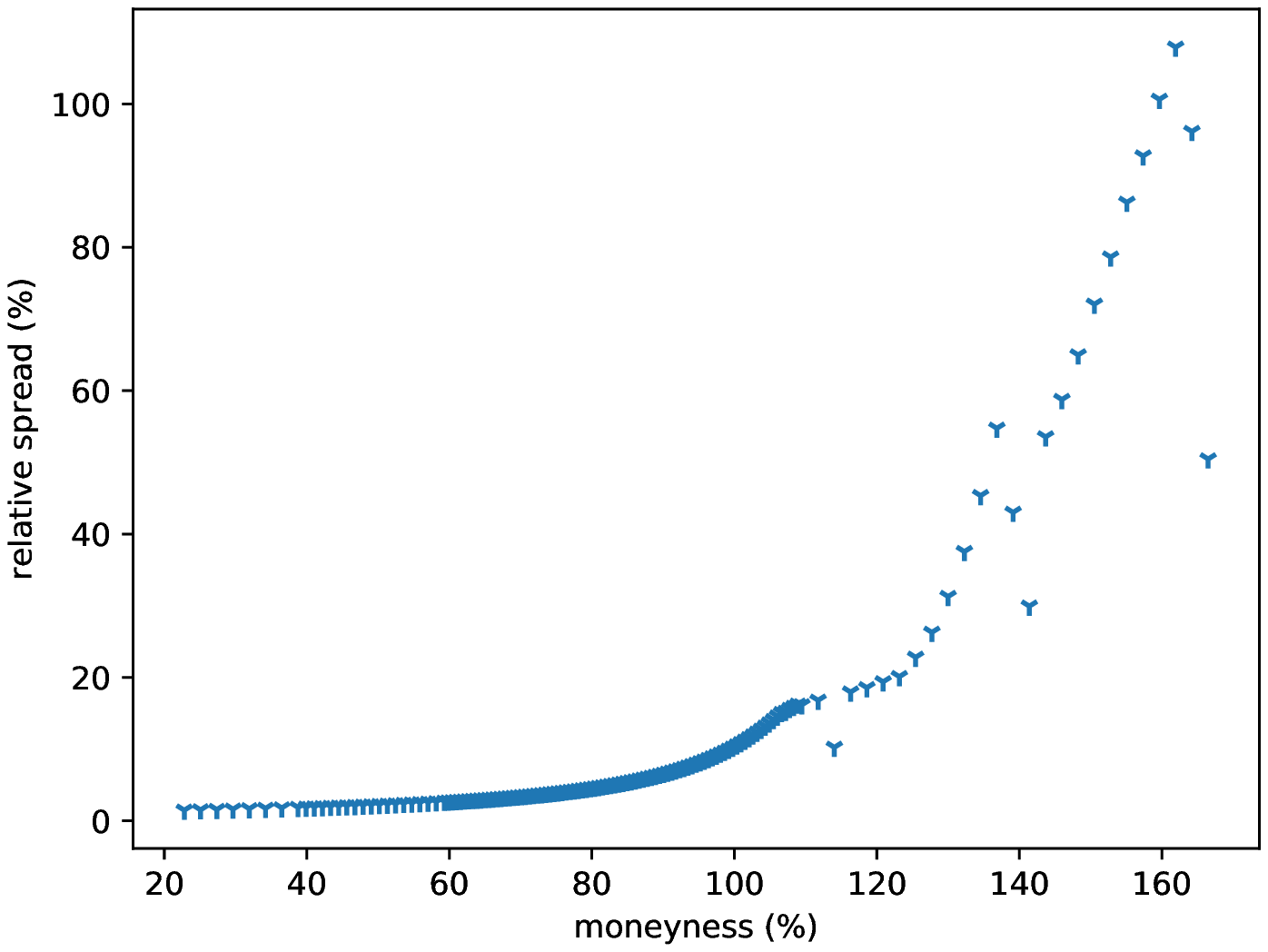}
     \caption{}\label{fig:SPX_relative_spread_call}
   \end{subfigure}
   \end{center}
   \caption{Absolute bid-ask spreads for the options considered in Figure \ref{fig:SPX}, panel (a), and their relative counterparts  (calculated with respect to mid prices), panel (b).}
\end{figure}

\begin{figure}[H]
\begin{center}
   \begin{subfigure}{.5\linewidth}
     \centering
     \includegraphics[scale=0.5]{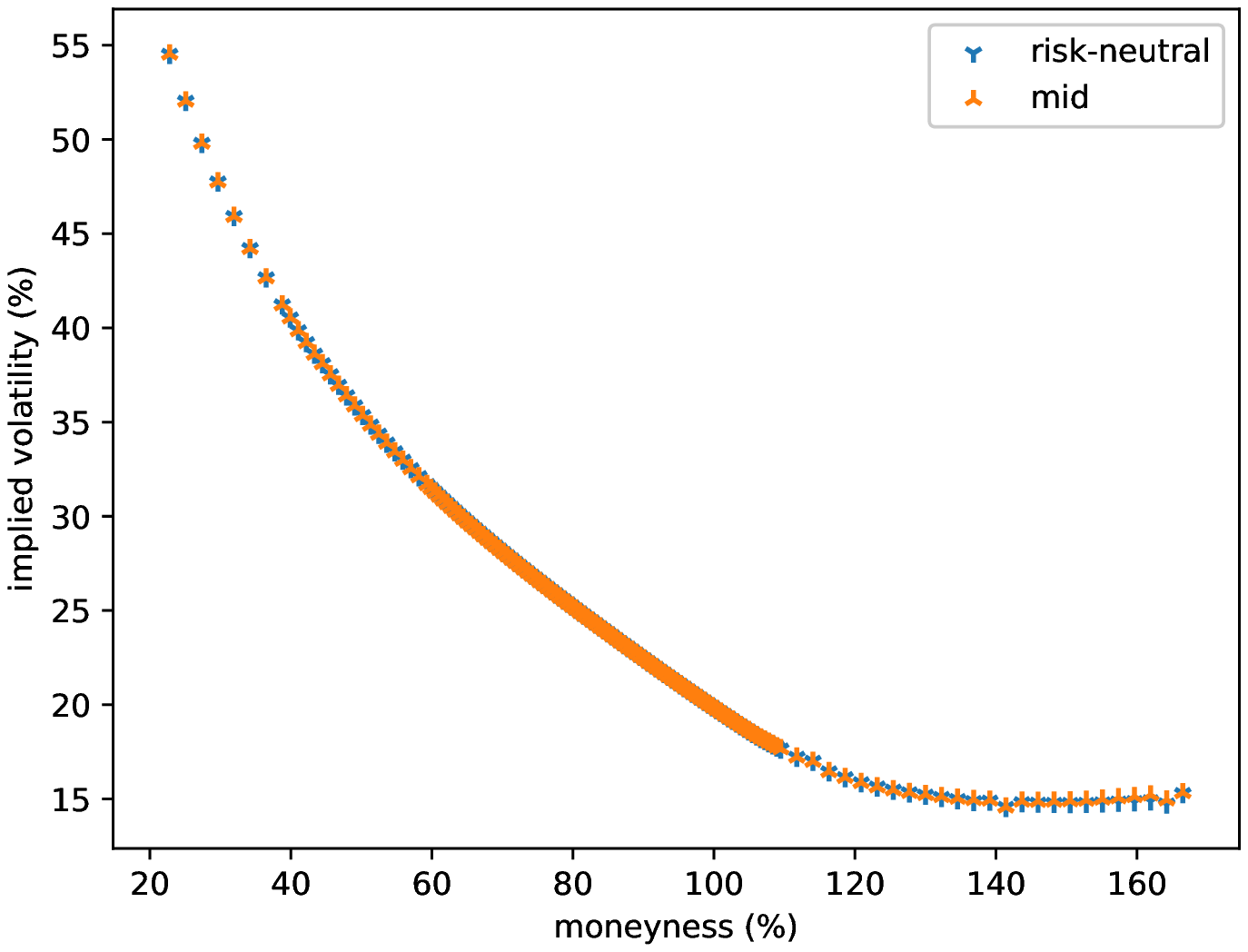}
     \caption{}\label{fig:SPX_implied_volatility_call}
   \end{subfigure}
   \begin{subfigure}{.5\linewidth}
    \centering
     \includegraphics[scale=0.5]{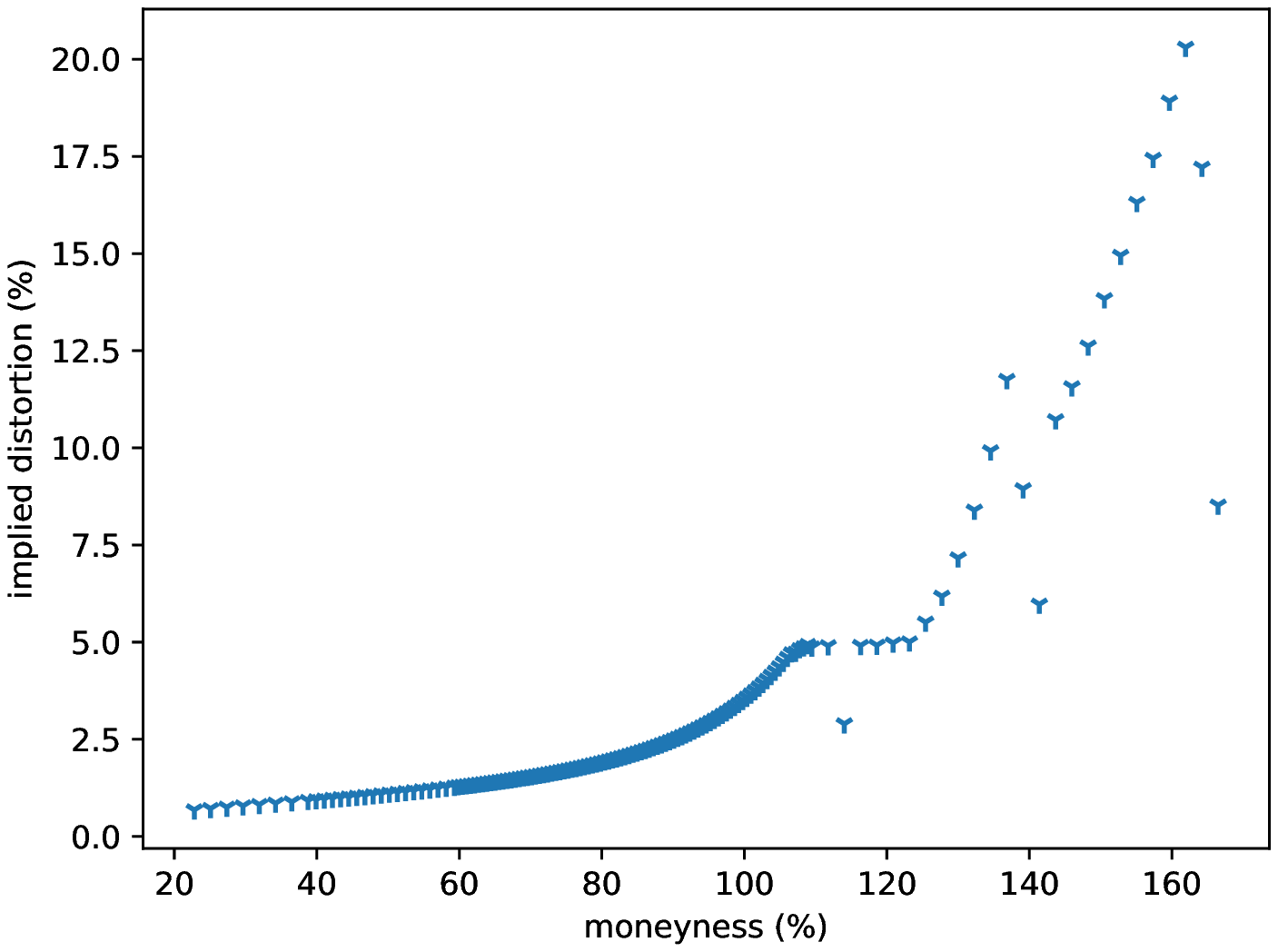}
     \caption{}\label{fig:SPX_implied_distortion_call}
   \end{subfigure}
   \end{center}
   \caption{Risk-neutral and mid implied volatilities, panel (a), and implied liquidity levels, panel (b), for the options considered in Figure \ref{fig:SPX}.}
\end{figure}


We now consider European put options on the FTSE MIB index. In this case fewer strikes are traded compared to the S\&P 500 case. However, as Figure \ref{fig:FTSE_bid_ask_put} illustrates,  these options are still very liquid; see also Figures \ref{fig:FTSE_absolute_spread_put} and \ref{fig:FTSE_relative_spread_put}. This is further confirmed by the low levels of the implied liquidity parameter of Figure \ref{fig:FTSE_implied_distortion_put}. We therefore still obtain risk-neutral implied volatilities and prices that are closely approximated by their mid counterparts; see Figures \ref{fig:FTSE_implied_volatility_put} and \ref{fig:FTSE_price_put}, respectively.

\begin{figure}[H]
\begin{center}
   \begin{subfigure}{.5\linewidth}
     \centering
     \includegraphics[scale=0.5]{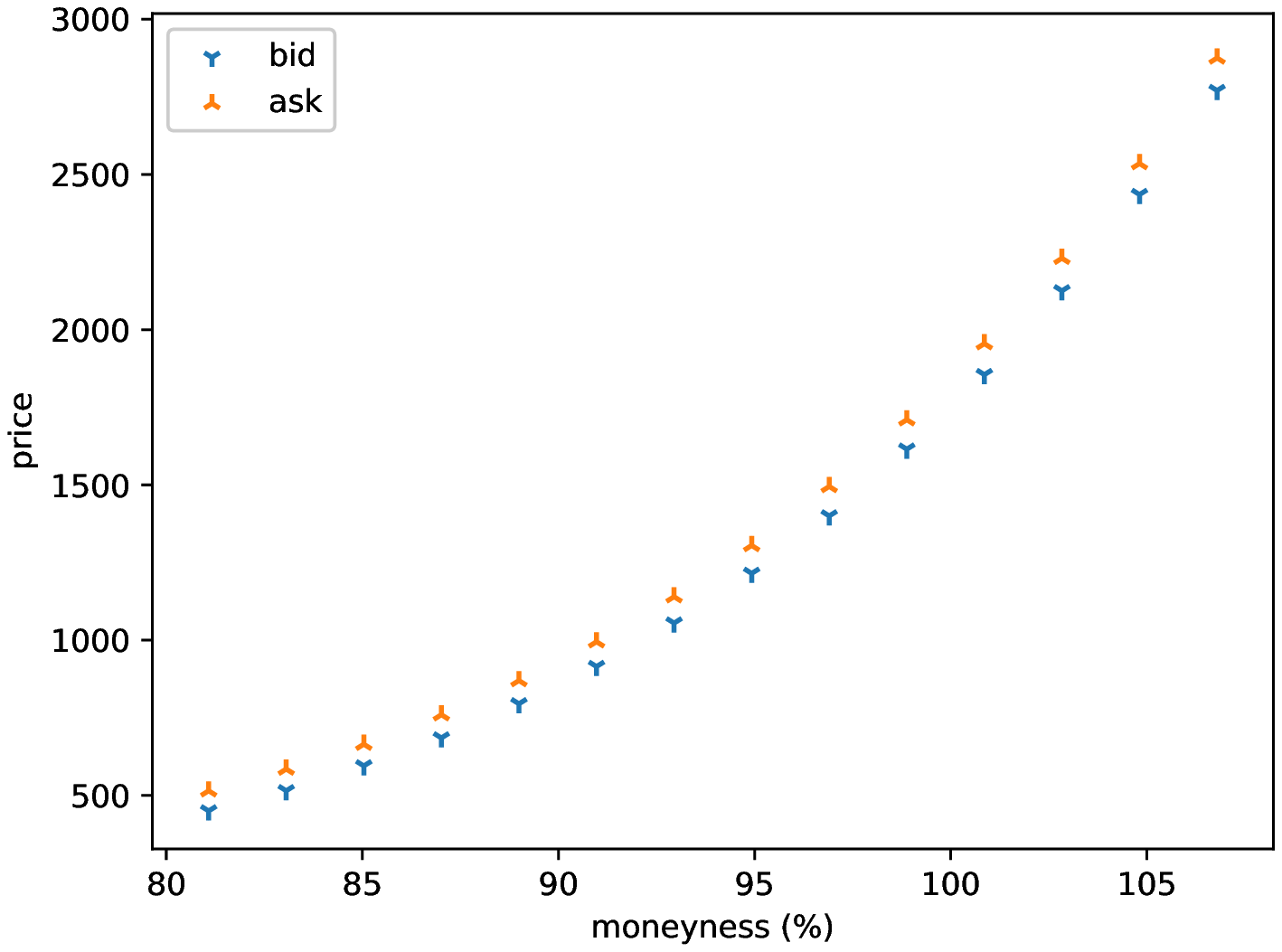}
     \caption{}\label{fig:FTSE_bid_ask_put}
   \end{subfigure}
   \begin{subfigure}{.5\linewidth}
    \centering
     \includegraphics[scale=0.5]{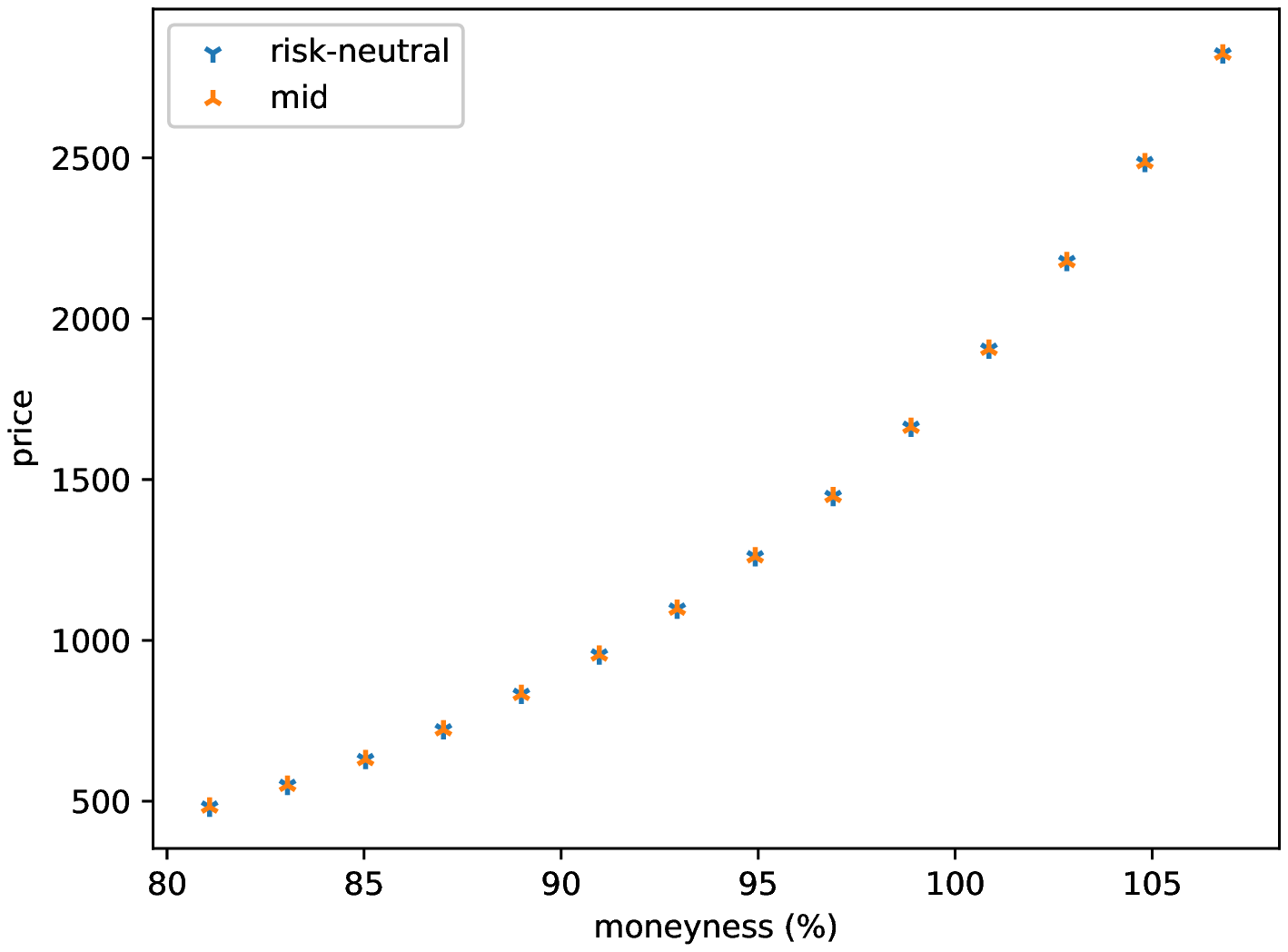}
     \caption{}\label{fig:FTSE_price_put}
   \end{subfigure}
   \end{center}
   \caption{Bid and ask prices for European put options on the FTSE MIB (Milan Stock Exchange) expiring in 249 days, panel (a), and their corresponding risk-neutral and mid counterparts, panel (b).\label{fig:FTSE}}
\end{figure}

\begin{figure}[H]
\begin{center}
   \begin{subfigure}{.5\linewidth}
     \centering
     \includegraphics[scale=0.5]{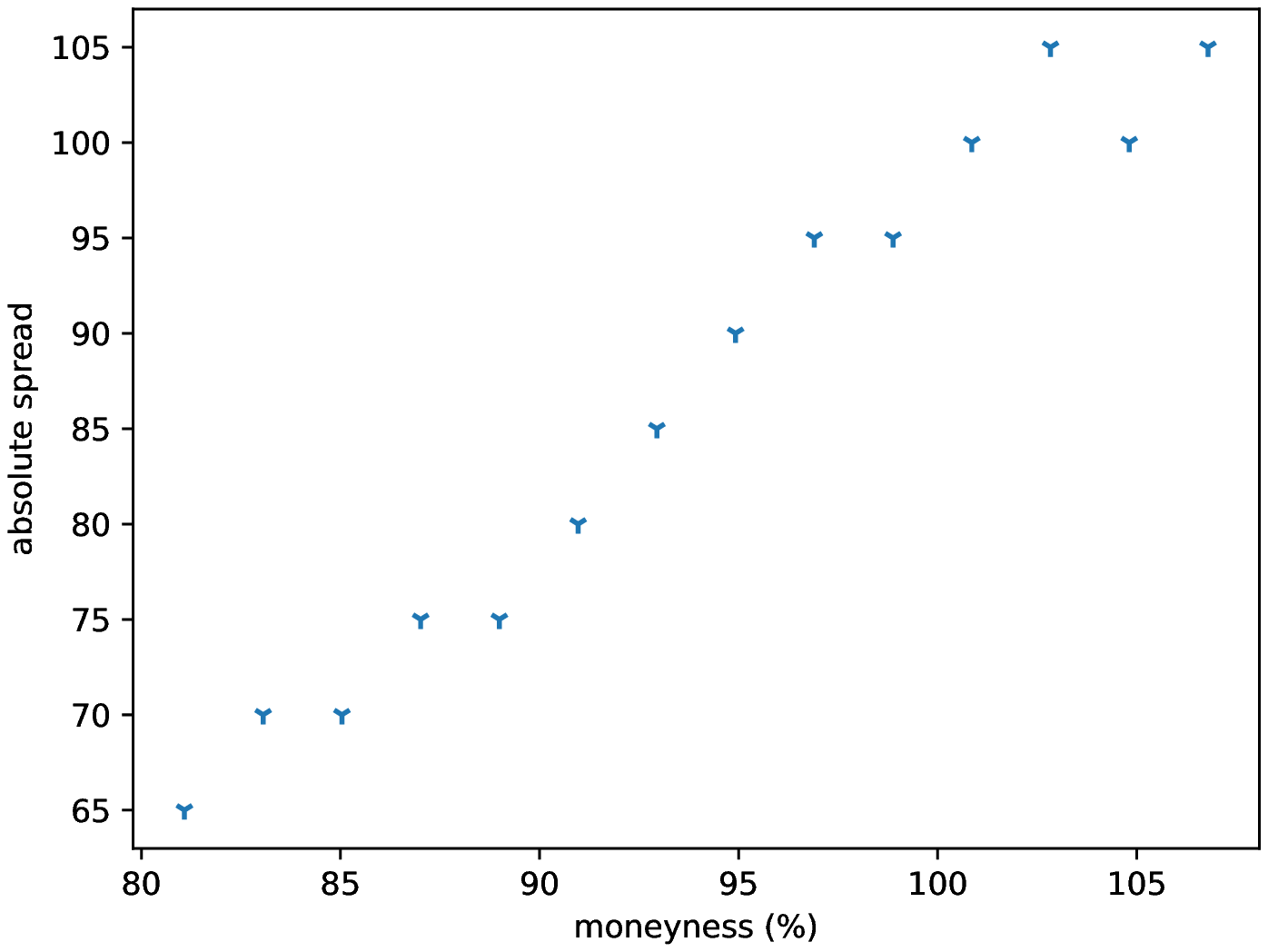}
     \caption{}\label{fig:FTSE_absolute_spread_put}
   \end{subfigure}
   \begin{subfigure}{.5\linewidth}
    \centering
     \includegraphics[scale=0.5]{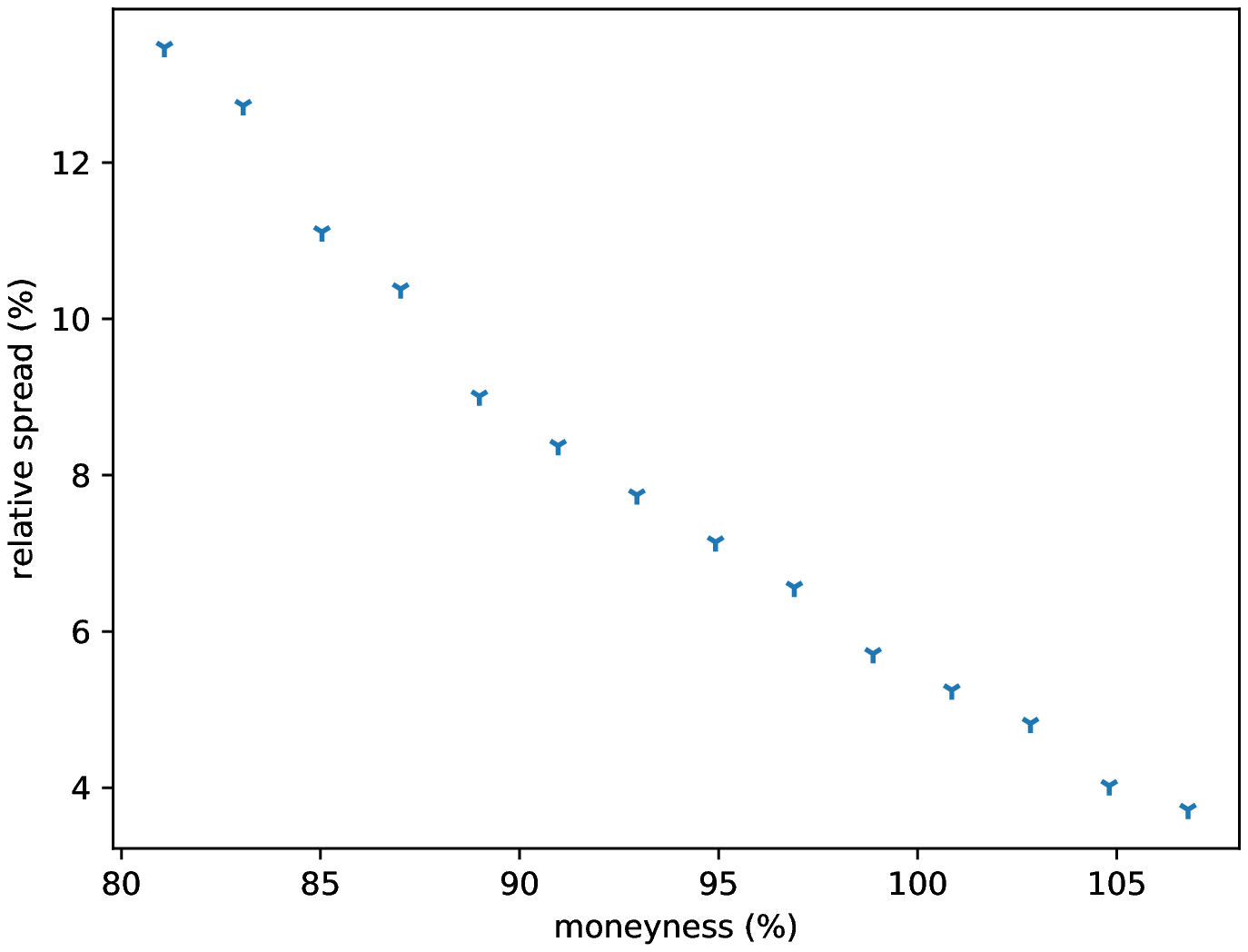}
     \caption{}\label{fig:FTSE_relative_spread_put}
   \end{subfigure}
   \end{center}
   \caption{Absolute bid-ask spreads for the options considered in Figure \ref{fig:FTSE}, panel (a), and their relative counterparts (calculated with respect to mid prices), panel (b).}
\end{figure}

\begin{figure}[H]
\begin{center}
   \begin{subfigure}{.5\linewidth}
     \centering
     \includegraphics[scale=0.5]{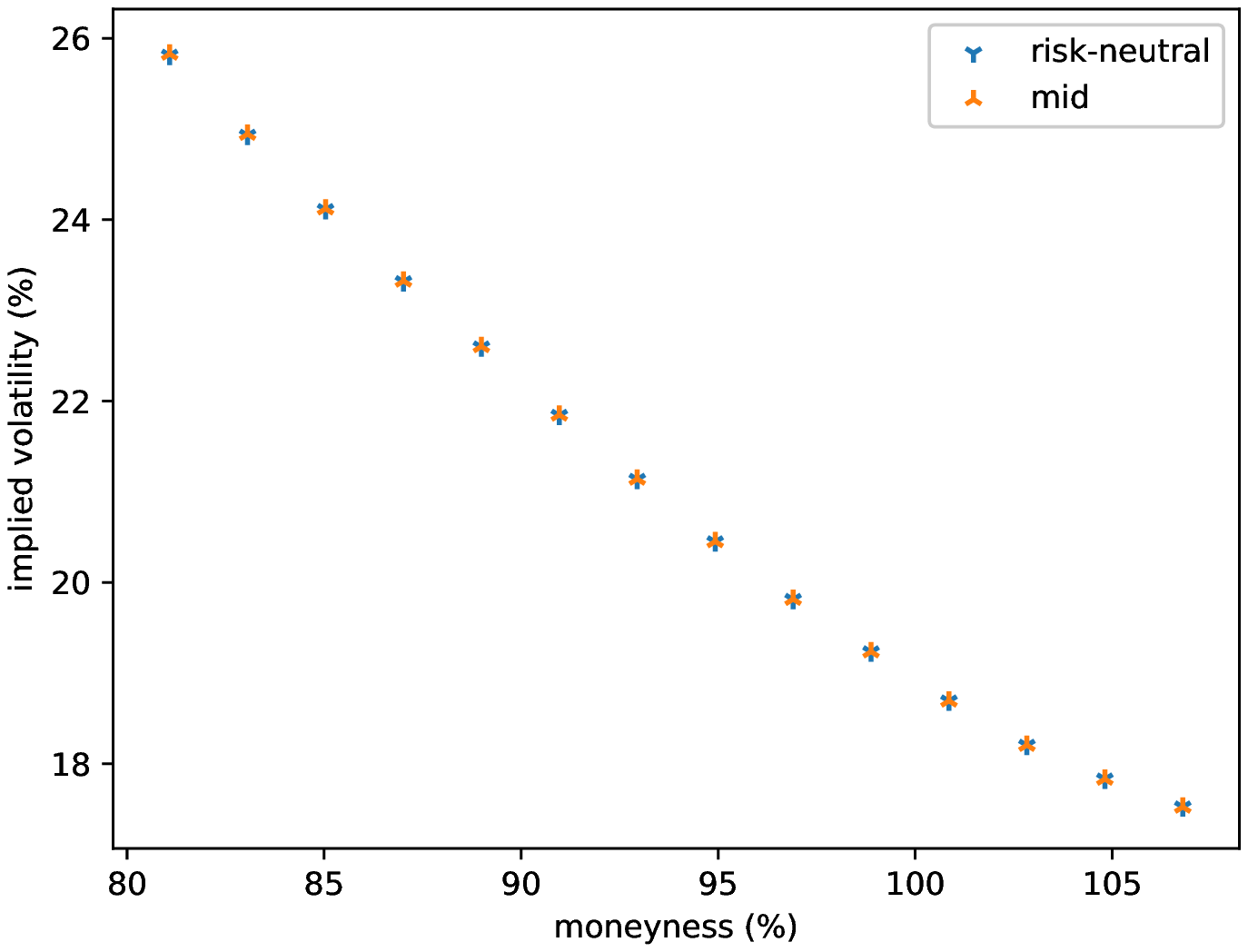}
     \caption{}\label{fig:FTSE_implied_volatility_put}
   \end{subfigure}
   \begin{subfigure}{.5\linewidth}
    \centering
     \includegraphics[scale=0.5]{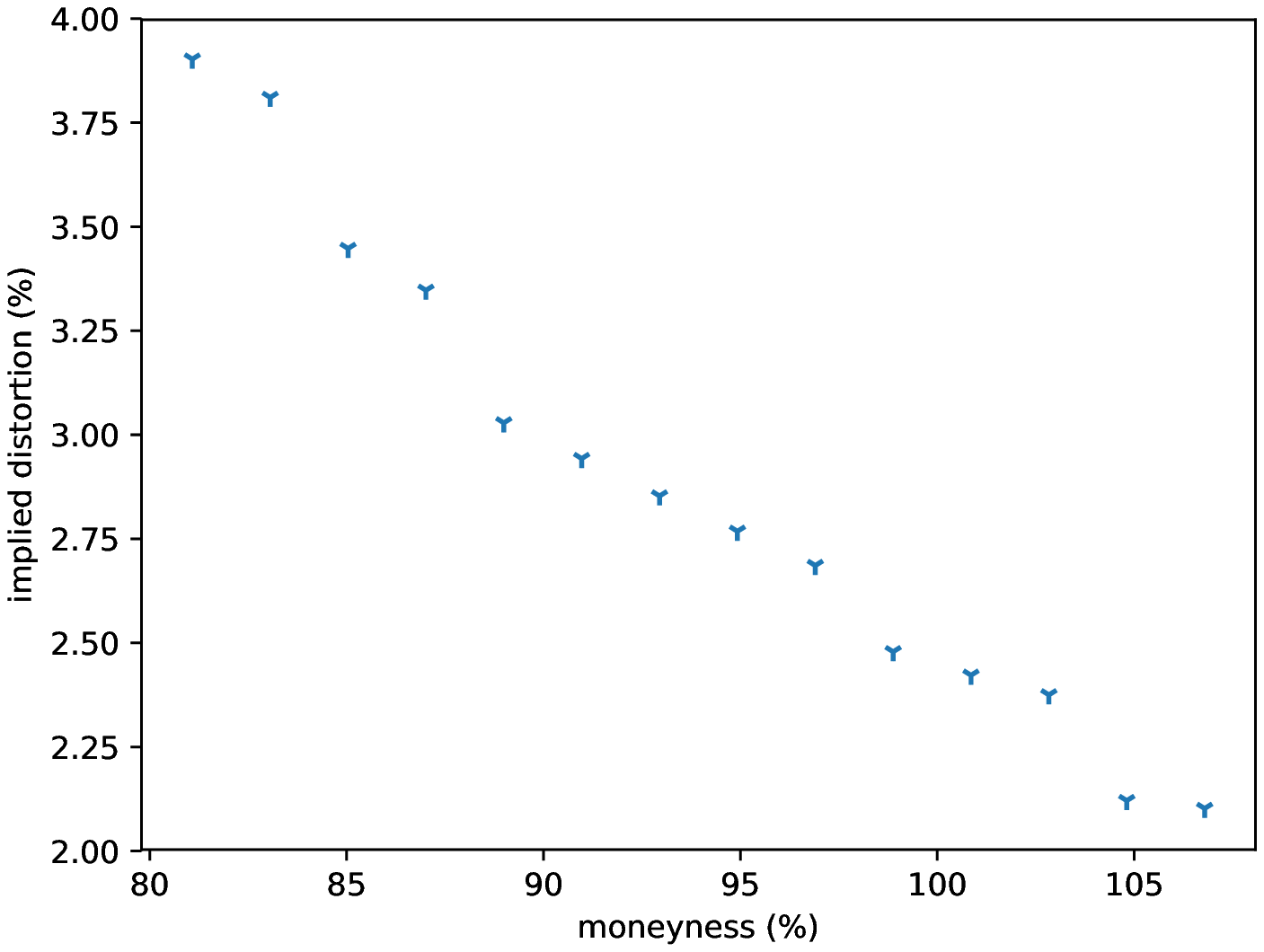}
     \caption{}\label{fig:FTSE_implied_distortion_put}
   \end{subfigure}
   \end{center}
   \caption{Risk-neutral and mid implied volatilities, panel (a), and implied liquidity levels, panel (b), for the options considered in Figure \ref{fig:FTSE}.}
\end{figure}

We now consider European call options on UBS. As it is clear from Figures \ref{fig:UBS_bid_ask_call}, \ref{fig:UBS_absolute_spread_call}, \ref{fig:UBS_relative_spread_call} and \ref{fig:UBS_implied_distortion_call}, these options are less liquid than those considered in the two cases above (i.e., those on the S\&P 500 and the FTSE MIB indices, respectively). Therefore, this results in risk-neutral and implied volatility smiles that, for deep out-of-the-money, but especially for deep in-the-money options, exhibit non-negligible differences, with mid implied volatilities overestimating their risk-neutral counterparts up to 2-3\% in the former case, and up to 9-10\% in the latter case; see Figure \ref{fig:UBS_implied_volatility_call}. Note that deep in-the-money and out-of-the-money options have small vegas, which leads to risk-neutral and mid option prices being close to each other; see figure \ref{fig:UBS_price_call}.

\begin{figure}[H]
\begin{center}
   \begin{subfigure}{.5\linewidth}
     \centering
     \includegraphics[scale=0.5]{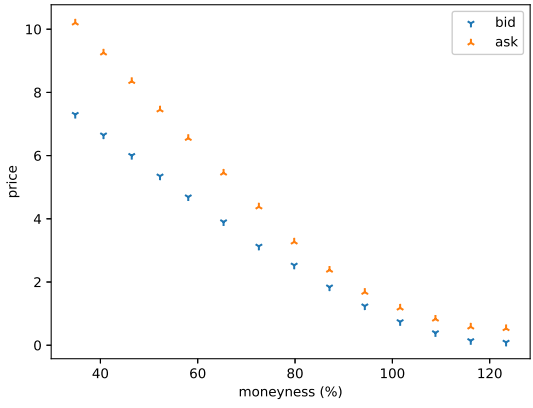}
     \caption{}\label{fig:UBS_bid_ask_call}
   \end{subfigure}
   \begin{subfigure}{.5\linewidth}
    \centering
     \includegraphics[scale=0.5]{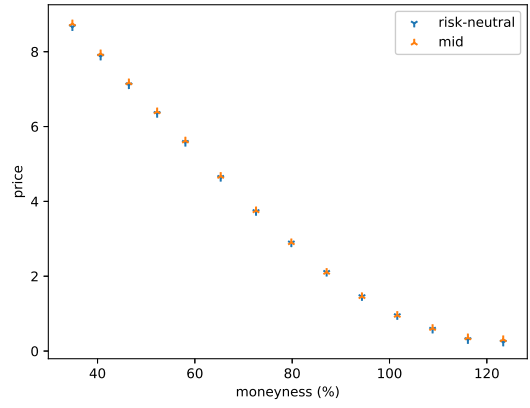}
     \caption{}\label{fig:UBS_price_call}
   \end{subfigure}
   \end{center}
   \caption{Bid and ask prices for European call options on UBS (Eurex) expiring in 345 days, panel (a), and their corresponding risk-neutral and mid counterparts, panel (b).\label{fig:UBS}}
\end{figure}

\begin{figure}[H]
\begin{center}
   \begin{subfigure}{.5\linewidth}
     \centering
     \includegraphics[scale=0.5]{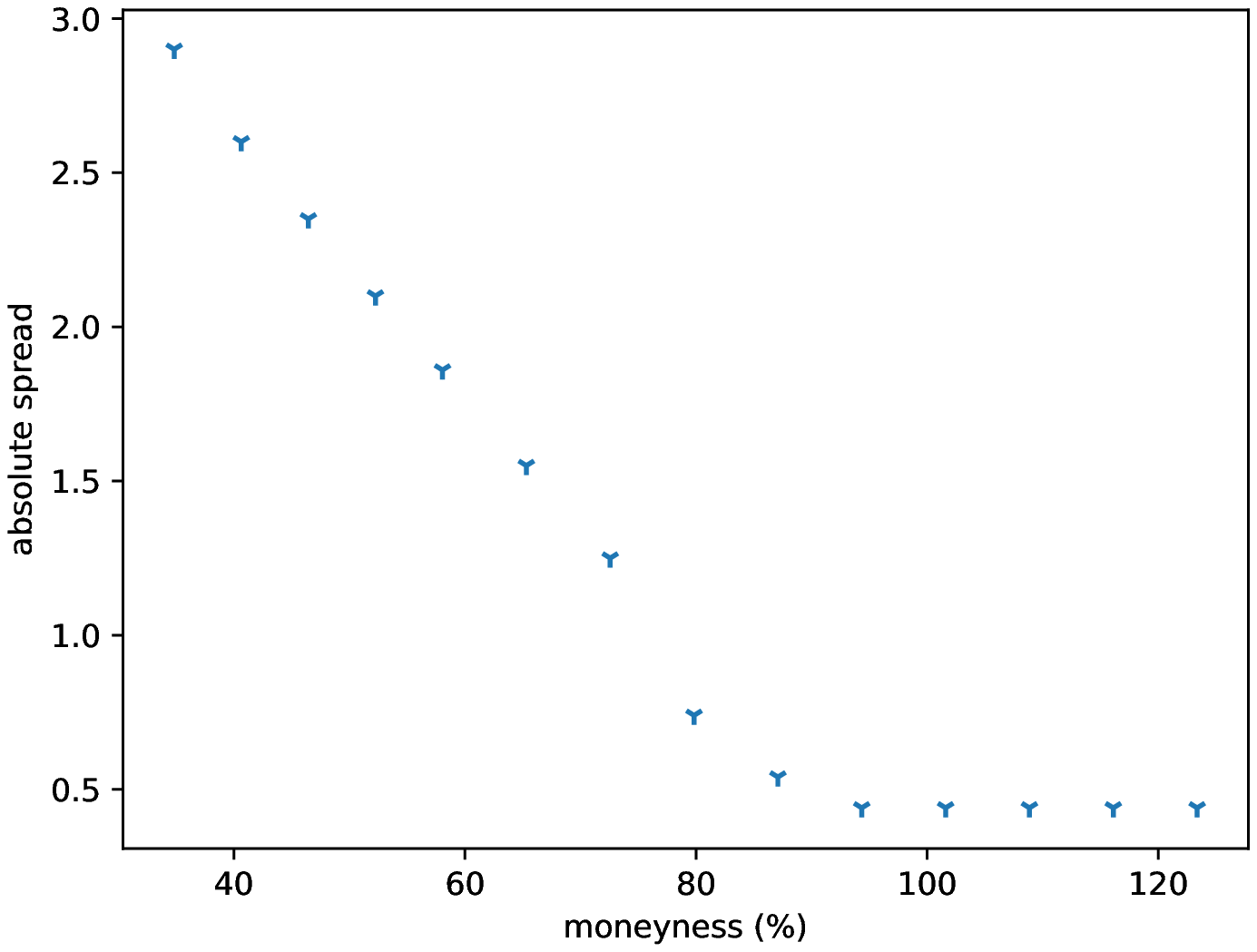}
     \caption{}\label{fig:UBS_absolute_spread_call}
   \end{subfigure}
   \begin{subfigure}{.5\linewidth}
    \centering
     \includegraphics[scale=0.5]{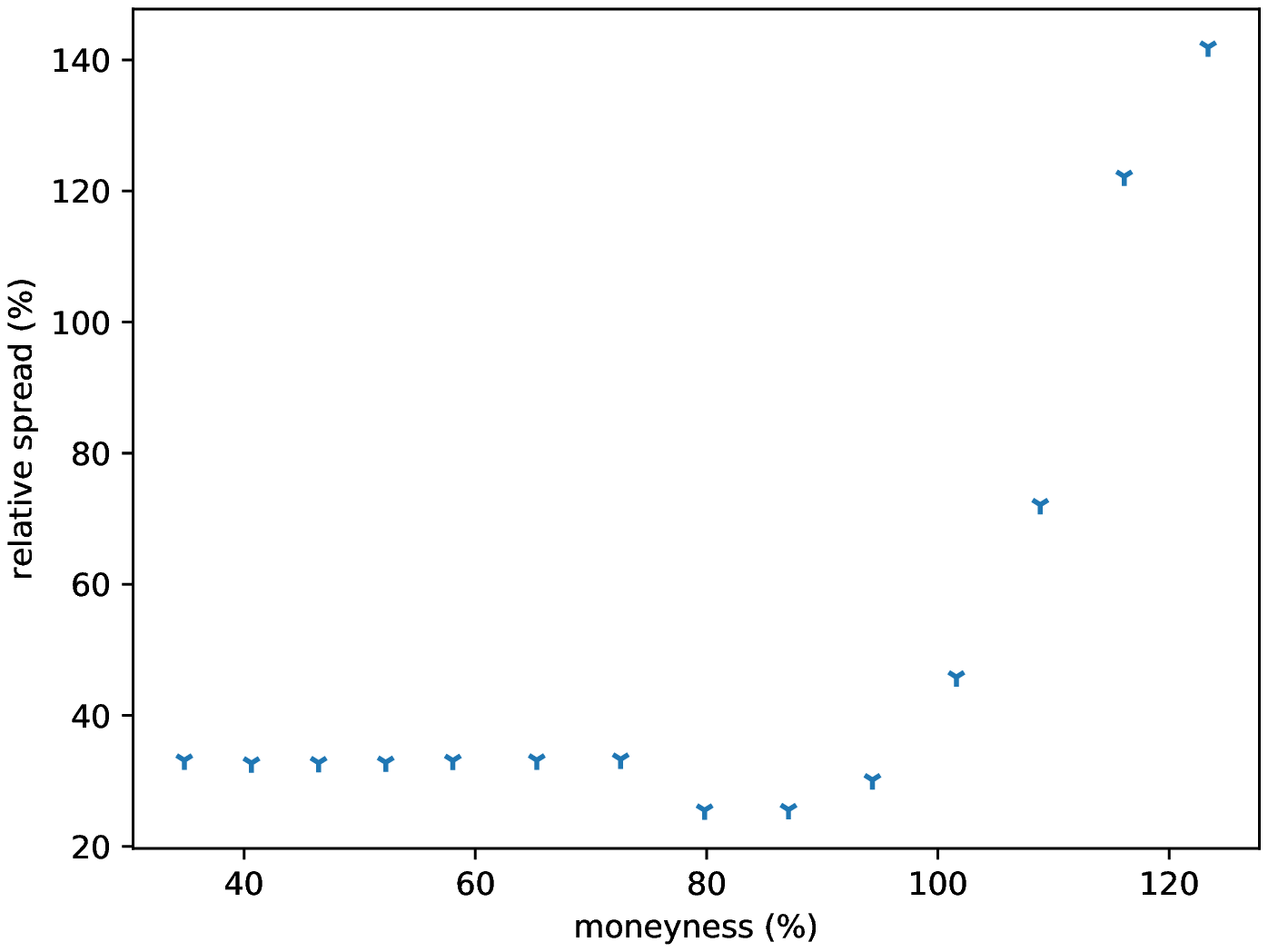}
     \caption{}\label{fig:UBS_relative_spread_call}
   \end{subfigure}
   \end{center}
   \caption{Absolute bid-ask spreads for the options considered in Figure \ref{fig:UBS}, panel (a), and their relative counterparts (calculated with respect to mid prices), panel (b).}
\end{figure}

\begin{figure}[H]
\begin{center}
   \begin{subfigure}{.5\linewidth}
     \centering
     \includegraphics[scale=0.5]{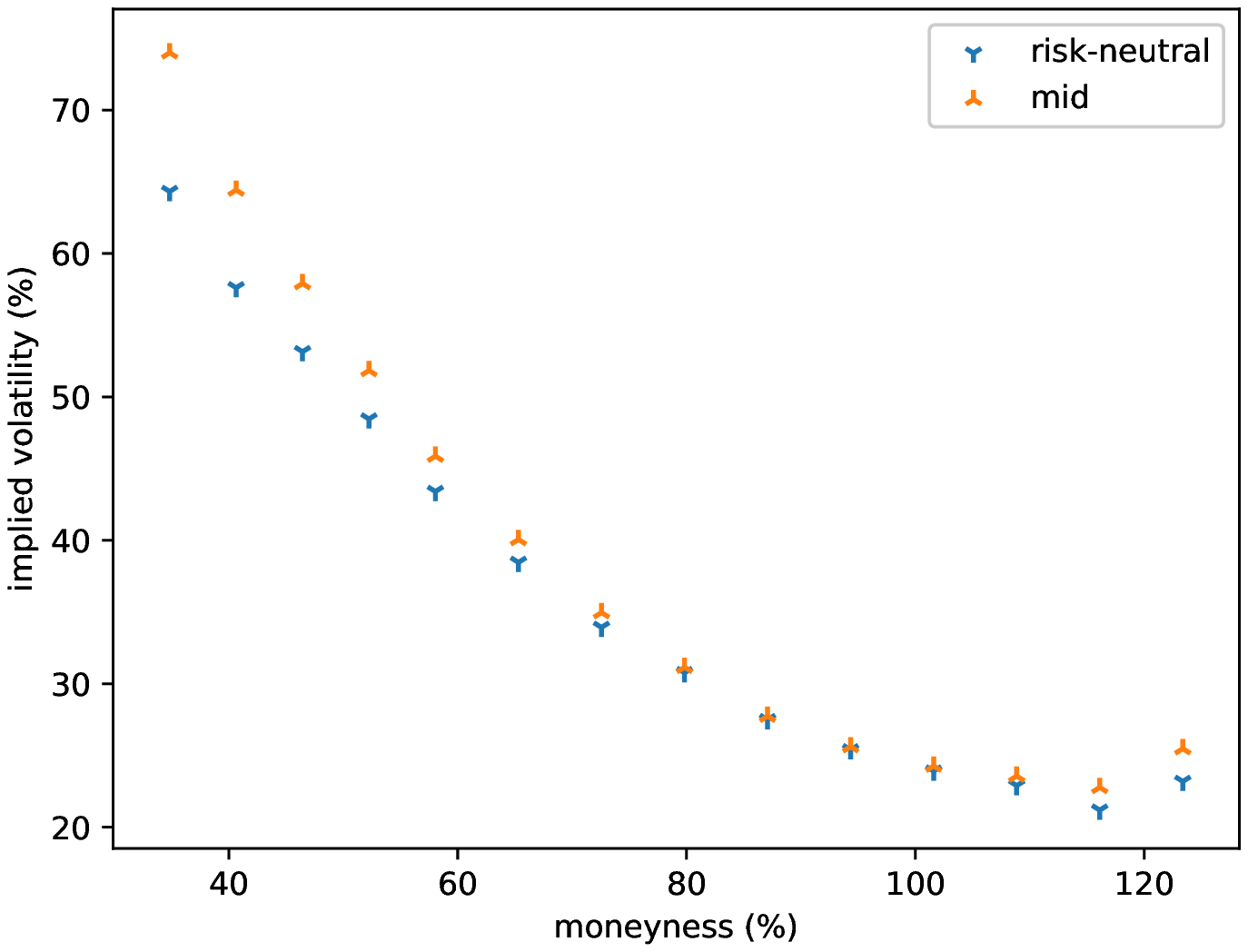}
     \caption{}\label{fig:UBS_implied_volatility_call}
   \end{subfigure}
   \begin{subfigure}{.5\linewidth}
    \centering
     \includegraphics[scale=0.5]{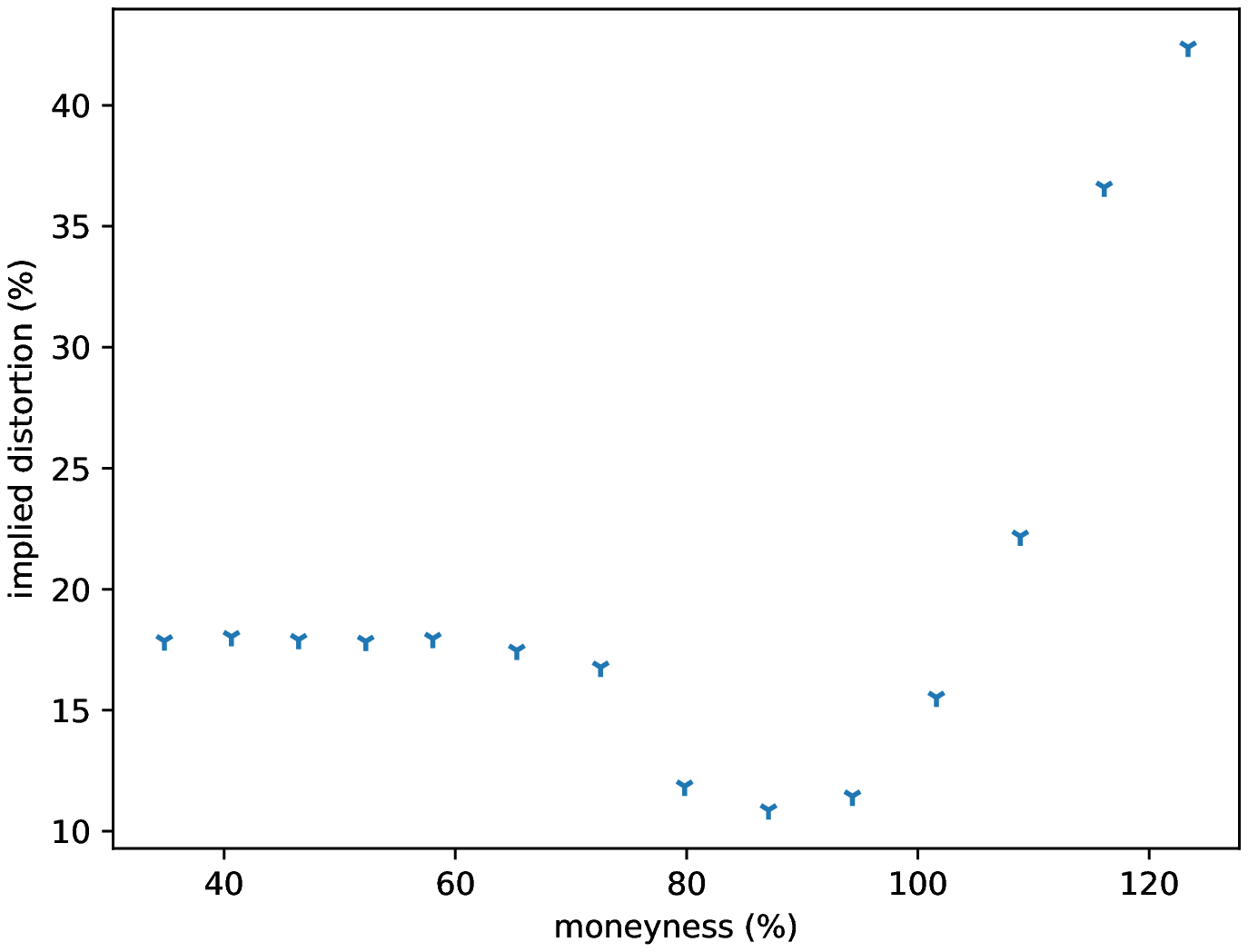}
     \caption{}\label{fig:UBS_implied_distortion_call}
   \end{subfigure}
   \end{center}
   \caption{Risk-neutral and mid implied volatilities, panel (a), and implied liquidity levels, panel (b), for the options considered in Figure \ref{fig:UBS}.}
\end{figure}

As the last case we consider that of European put options on Deutsche Telekom; see Figure \ref{fig:DTEG_bid_ask_put}. Also in these circumstances liquidity is not as high as in the cases of options on the S\&P 500 and FTSE MIB indices. This is illustrated by the high levels of bid-ask spreads displayed in Figures \ref{fig:DTEG_absolute_spread_put} and \ref{fig:DTEG_relative_spread_put}, and reiterated by the high implied liquidity levels of Figure \ref{fig:DTEG_implied_distortion_put}. Due to the low liquidity for both in-the-money and out-of-the-money options for this particular underlying, differences between risk-neutral and mid implied volatilities are considerable; see Figures \ref{fig:DTEG_implied_volatility_put}. In the former case we observe mid implied volatilities underestimating their risk-neutral counterparts up to 9-10\%, while in the latter case mid implied volatilities overestimate risk-neutral ones, with differences up to 14-15\%. Also in this case mid prices are good proxies for their risk-neutral counterparts; see Figure \ref{fig:DTEG_price_put}. This is again due to the fact that close to the at-the-money point liquidity is high, and far from it, even if liquidity decreases, options are not very sensitive to volatility changes.

\begin{figure}[H]
\begin{center}
   \begin{subfigure}{.5\linewidth}
     \centering
     \includegraphics[scale=0.5]{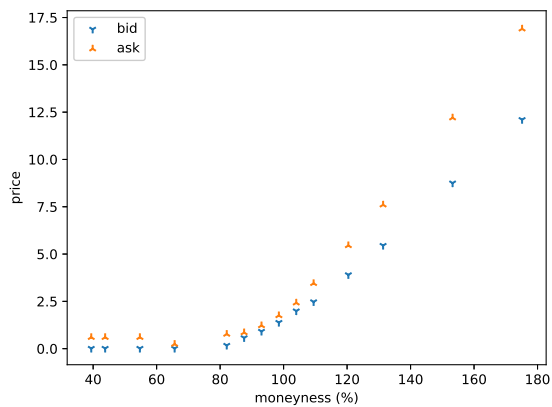}
     \caption{}\label{fig:DTEG_bid_ask_put}
   \end{subfigure}
   \begin{subfigure}{.5\linewidth}
    \centering
     \includegraphics[scale=0.5]{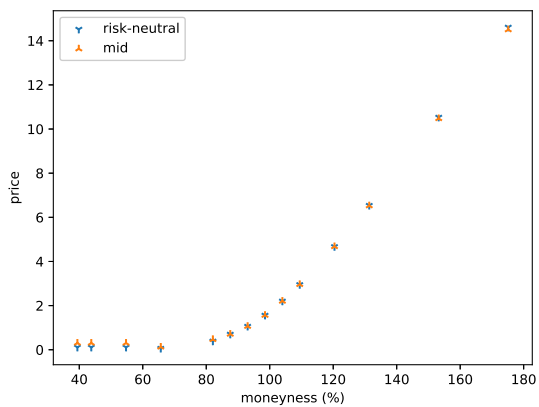}
     \caption{}\label{fig:DTEG_price_put}
   \end{subfigure}
   \end{center}
   \caption{Bid and ask prices for European put options on Deutsche Telekom (Eurex) expiring in 345 days, panel (a), and their corresponding risk-neutral and mid counterparts, panel (b).\label{fig:DTEG}}
\end{figure}

\begin{figure}[H]
\begin{center}
   \begin{subfigure}{.5\linewidth}
     \centering
     \includegraphics[scale=0.5]{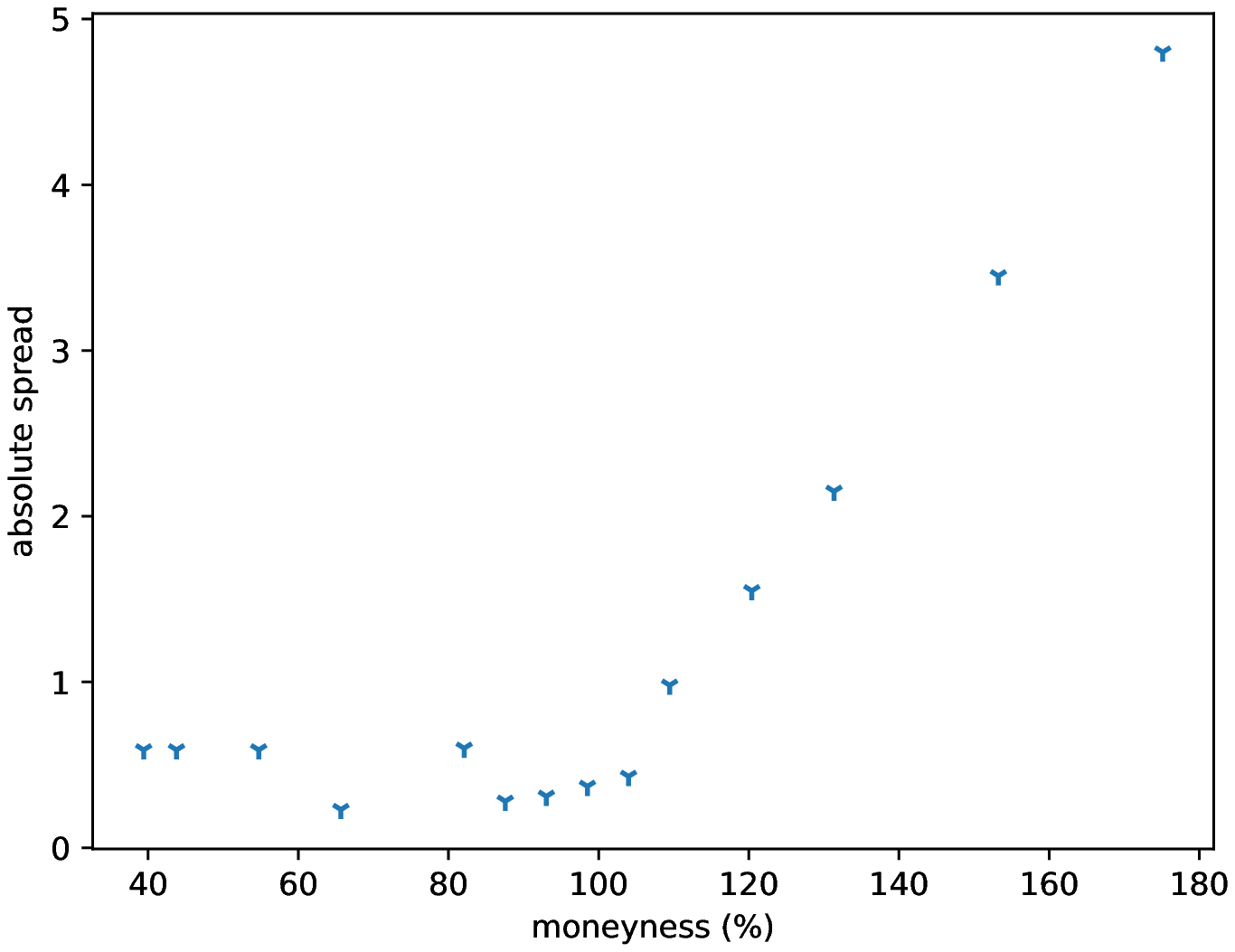}
     \caption{}\label{fig:DTEG_absolute_spread_put}
   \end{subfigure}
   \begin{subfigure}{.5\linewidth}
    \centering
     \includegraphics[scale=0.5]{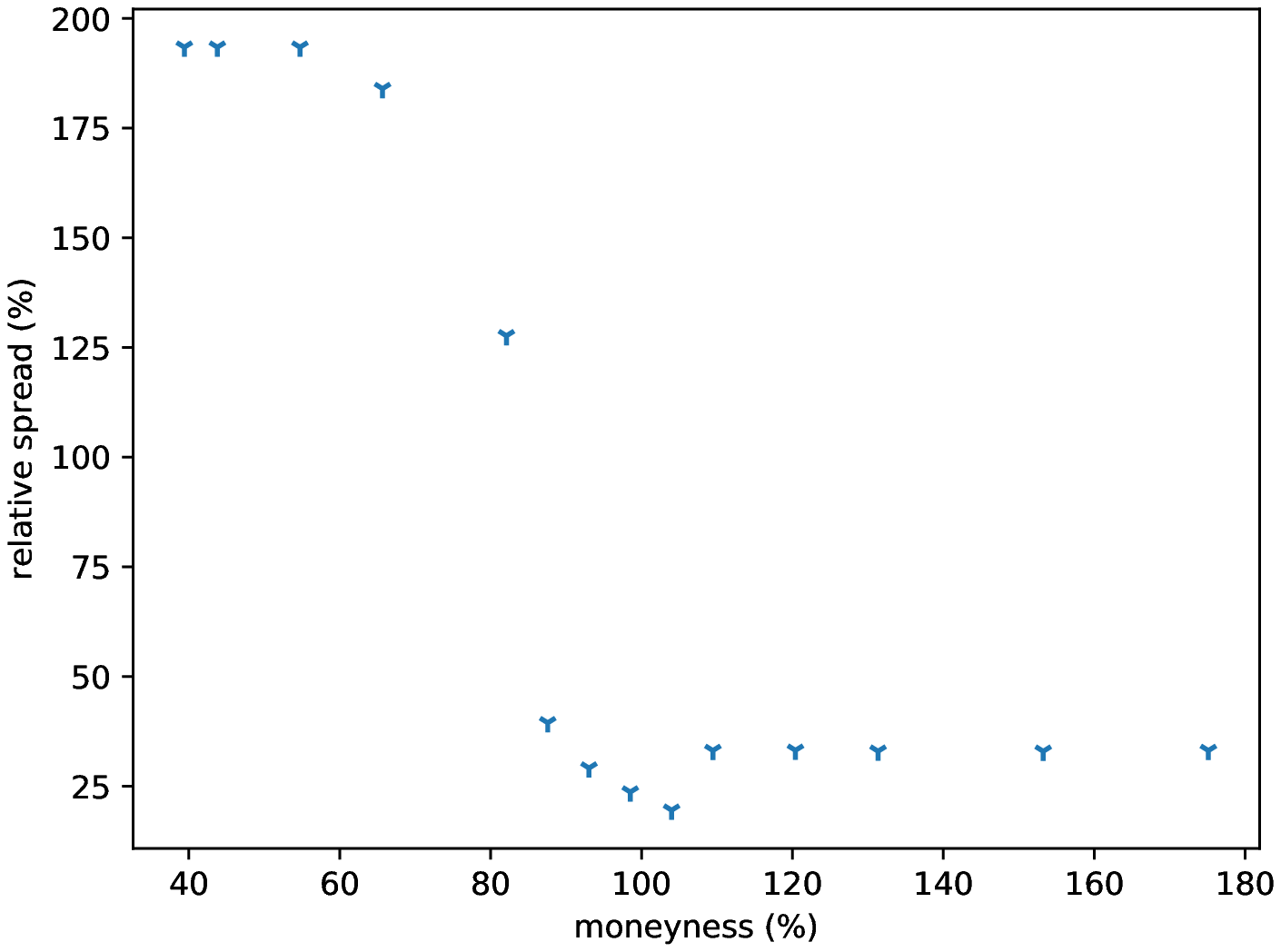}
     \caption{}\label{fig:DTEG_relative_spread_put}
   \end{subfigure}
   \end{center}
   \caption{Absolute bid-ask spreads for the options considered in Figure \ref{fig:DTEG}, panel (a), and their relative counterparts (calculated with respect to mid prices), panel (b).\label{fig:DTEG_spreads}}
\end{figure}

\begin{figure}[H]
\begin{center}
   \begin{subfigure}{.5\linewidth}
     \centering
     \includegraphics[scale=0.5]{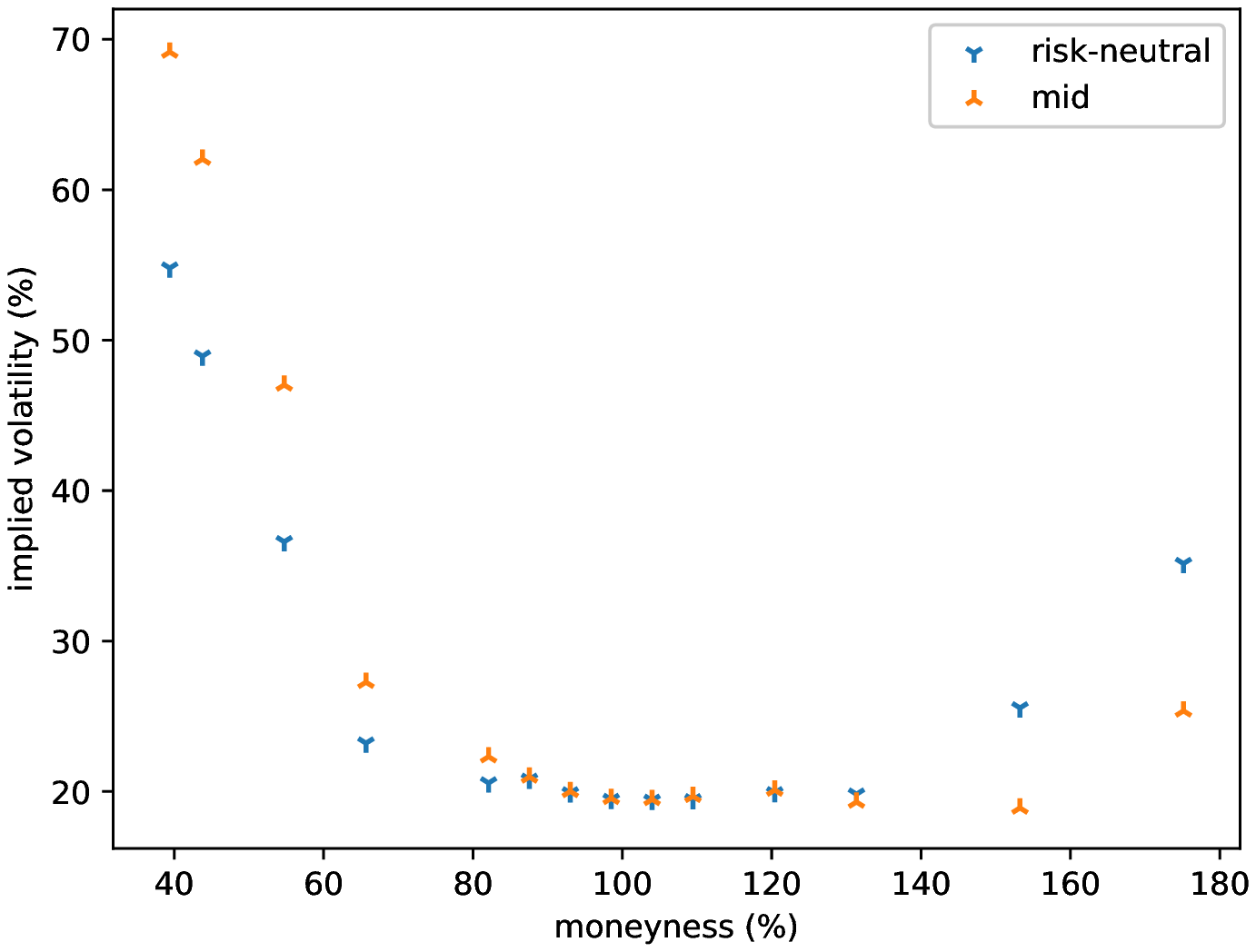}
     \caption{}\label{fig:DTEG_implied_volatility_put}
   \end{subfigure}
   \begin{subfigure}{.5\linewidth}
    \centering
     \includegraphics[scale=0.5]{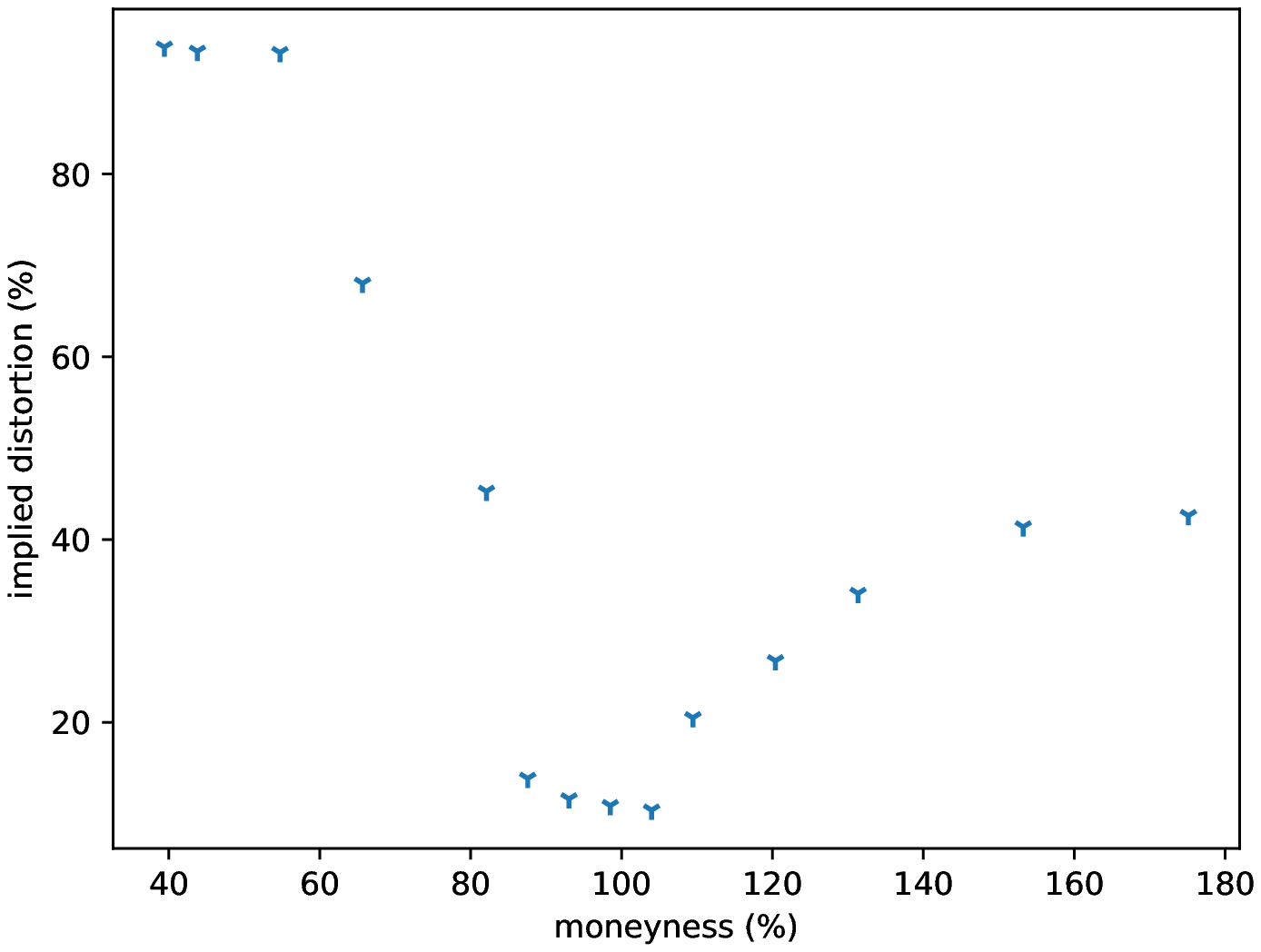}
     \caption{}\label{fig:DTEG_implied_distortion_put}
   \end{subfigure}
   \end{center}
   \caption{Risk-neutral and mid implied volatilities, panel (a), and implied liquidity levels, panel (b), for the options considered in Figure \ref{fig:DTEG}.}
\end{figure}

Overall, we see that for very liquid instruments mid implied volatilities are very well approximated by their risk-neutral counterparts, as expected. When liquidity decreases, however, for in-the-money and out-of-the-money options implied risk-neutral volatilities might differ noticeably from mid volatilities. Predictably, this does not have a considerable impact on option prices. This is because close to the at-the-money point liquidity is in general high, making risk-neutral and mid implied volatilities close to each other. On the other hand, for in-the-money and out-of-the-money options liquidity can considerably affect volatilities. However, these options have low vegas, which makes their prices not very sensitive to changes in the volatility of the underlying. Nonetheless, whether risk-neutral and mid prices are close to each other is beside the point: for both risk-neutral and mid prices there is no liquidity in the market, so from a trading perspective only bid and ask prices matter. What we are interested in is assessing whether implied volatilities can be extracted from traded option quotes in a consistent manner with the risk-neutral framework. In particular, what we believe is important is to assess how considering bid and ask prices as a starting point instead of their mid counterparts can affect the shape of the volatility smile. As we have seen in the examples considered, computing implied volatilities from bid and ask prices instead of from mid prices in some cases can have a large impact on the implied volatility figures, and this can have implications in different contexts. As an example, if a smile model is calibrated by means of a least-square approach to the available implied volatilities, then differences as those observed in Figures \ref{fig:UBS_implied_volatility_call} and \ref{fig:DTEG_implied_volatility_put} would result in risk-neutral and mid volatility smiles with different shapes, as in-the-money and out-of-the-money implied volatilities would affect the calibration as a whole. Furthermore, when simulation models for over-the-counter derivatives, as for instance credit models, calibrated to implied volatilities are used, then choosing to input risk-neutral rather than mid implied volatilities might have a non-marginal impact for running contracts which, trough time, due to market movements ended up being in-the-money or out-of-the-money. Therefore, the examples considered, as well as the theoretical consistency of the methodology outlined in this article with the risk-neutral paradigm (paradigm that is not satisfied when mid prices are considered), make liquidity-free implied implied volatilities a potentially useful tool in financial modeling.

\section{Conclusion}\label{sec:conclusion}
In this article we have considered the problem of computing implied volatilities from bid and ask European option prices directly, i.e., without relying on mid price approximations. The methodology we have outlined relies on the conic finance framework of \cite{cherny2009}. Based on the results of \cite{michielon} it is possible, given the bid and ask prices of an option, to imply both the risk-neutral volatility and the liquidity level of the market at the same time. In particular, in the case of Black-Scholes and Bachelier specifications, this procedure results particularly efficient when the Wang transform is used, as the latter allows to analytically compute bid and ask option prices. In the case of the Bachelier model, these analytical formulae have been provided. The methodology outlined in this article relies on some intuitive and simple assumptions concerning the liquidity level of the market and the wideness of the range of option prices with respect to changes in the volatility parameter. A potential application for the technique we propose is that of constructing liquidity-free implied volatilitiy surfaces (and, consequently, corresponding implied liquidity surfaces at the same time). These liquidity-free implied volatility surfaces could be used as calibration inputs for different models under risk-neutral settings in a consistent manner.

\section*{Disclaimer}
The opinions expressed in this paper are solely those of the authors and do not necessarily reflect those of their current and past employers.

\clearpage

\bibliographystyle{apalike}
\bibliography{biblio}

\clearpage
\appendix

\section{Proof of Theorem \ref{th:general}\label{sec:proof}}
In order to prove Theorem \ref{th:general}, we provide two technical lemmas. The results presented here are based on \cite{michielon}, and they have been made available here to make the article self-contained. If we assume that the quoted bid and ask prices of the contingent claim $Y$ lie within the interval of all the possible risk-neutral prices that can be obtained by changing the parameter $\lambda$, and that for every $\lambda$ in a given interval it is always possible to find a distortion parameter $\gamma$ such that the observed bid-ask spread can be reproduced, Lemma \ref{lemma:lambda interval} follows.

\begin{lemma}\label{lemma:lambda interval}
Let $Y$ be a contingent claim whose price depends on a parameter $\lambda>0$ such that the risk-neutral price of $Y$ is strictly increasing with respect to $\lambda$. Assume that the inequalities $\inf_{\lambda> 0} \PV{Y(\lambda)} < b$ and $\sup_{\lambda> 0} \PV{Y(\lambda)} > a$ hold. Then, there exists an interval $[\lambda_b, \lambda_a]$ such that there is equivalence between $b\leq \PV{Y(\lambda)} \leq a$ and $\lambda \in [\lambda_b, \lambda_a]$. Further, assume that for every $\lambda \in [\lambda_b, \lambda_a]$ there exists $\gamma> 0$ such that $\ask{Y(\lambda),\gamma} - \bid{Y(\lambda),\gamma} = a-b$. Then, such $\gamma$ is unique.
\end{lemma}
\begin{proof}
The existence of the interval $[\lambda_b, \lambda_a]$ immediately follows from $\PV{Y(\lambda)}$ being an increasing and continuous function of $\lambda$.

We now consider a fixed $\lambda \in [\lambda_b, \lambda_a]$. There exists at least one $\gamma>0$ such that $\ask{Y(\lambda),\gamma}-\bid{Y(\lambda),\gamma} = a-b$. To derive a contradiction, suppose that there exist $\gamma_*$ and $\gamma^*$ satisfying the relationships $\ask{Y(\lambda),\gamma_*}-\bid{Y(\lambda),\gamma_*}=\ask{\lambda,\gamma^*}-\bid{\lambda,\gamma^*}=a-b$, where $\gamma_*<\gamma^*$. We recall that the ask (bid) price of $Y(\lambda)$ is an increasing (decreasing) function of $\gamma$. From this it follows that $\ask{\lambda,\gamma_*}<\ask{\lambda,\gamma^*}$, and that $\bid{\lambda,\gamma_*}>\bid{\lambda,\gamma^*}$. Therefore, we obtain that $a-b=\ask{\lambda,\gamma_*}-\bid{\lambda,\gamma_*}<\ask{\lambda,\gamma^*}-\bid{\lambda,\gamma^*}=a-b$, which leads to a contradiction.
\end{proof}

We now prove the following continuity-related result.
	
\begin{lemma}\label{lemma:continuity gamma}
Under the hypotheses of Lemmas \ref{lemma:lambda interval}, for every $\lambda \in [\lambda_b, \lambda_a]$ the function such that $\lambda\mapsto \gamma(\lambda)$, where $\mathrm{ask}(Y(\lambda), \gamma(\lambda))-\mathrm{bid}(Y(\lambda), \gamma(\lambda))=a-b$, is continuous.
\end{lemma}
\begin{proof}
We consider $\bar{\lambda}$ in $[\lambda_b, \lambda_a]$ fixed, as well as a sequence $(\lambda_n)_n$ in $[\lambda_b, \lambda_a]$ converging to $\bar{\lambda}$. Further, we set $\phi(\lambda,\gamma)\coloneqq\ask{Y(\lambda),\gamma} - \bid{Y(\lambda),\gamma}$. We proceed in three steps.

(i) We first prove that the sequence $(\gamma(\lambda_n))_n$ is bounded. By contradiction, assume this is not the case. It then follows that there exists a subsequence $(\gamma(\lambda_{n_k}))_k$ of $(\gamma(\lambda_n))_n$ diverging to $+\infty$. The sequence $(\lambda_{n_k})_k$ converges to $\bar{\lambda}$, as it is a subsequence of a convergent sequence. Further, recall that $\phi$ is continuous in both arguments. Therefore, $\lim_k \phi(\bar{\lambda}_{n_k}, \gamma(\lambda_{n_k}))=\phi(\bar{\lambda}, +\infty)=a-b$, as $\phi(\bar{\lambda}_{n_k}, \gamma(\lambda_{n_k}))$  always equals $a-b$, by construction. As assumed in Lemma \ref{lemma:lambda interval} there exists $\bar{\gamma}> 0$ such that $\phi(\bar{\lambda}, \bar{\gamma})=a-b$. Due to $\phi$ being increasing in its second argument, we obtain that $a-b = \phi(\bar{\lambda}, \bar{\gamma})<\phi(\bar{\lambda}, +\infty)=a-b$, contradiction.

(ii) We now prove that the sequence $(\gamma(\lambda_n))_n$ admits limit. As $(\gamma(\lambda_n))_n$ is bounded, it has a convergent subsequence. Assume that there exist two subsequences, denoted as $(\gamma(\lambda_{n_k}))_k$ and $(\gamma(\lambda_{n_h}))_h$, that converge to $\gamma_*$ and $\gamma^*$, respectively, with $\gamma_*<\gamma^*$. The sequences $(\lambda_{n_k})_k$ and $(\lambda_{n_h})_h$ are subsequencies of the same convergent sequence. Therefore, they both converge to $\bar{\lambda}$. We then obtain that $a-b = \lim_k\phi(\lambda_{n_k},\gamma(\lambda_{n_k}))=\phi(\bar{\lambda},\gamma_*) <\phi(\bar{\lambda},\gamma^*) = \lim_h \phi(\lambda_{n_h},\gamma(\lambda_{n_h}))=a-b$, contradiction (the first and the last equalities are due to the definitions of $(\lambda_{n_k})_k$ and $(\lambda_{n_h})_h$, respectively, the second and the penultimate equalities follow from the continuity of $\phi$, while the inequality is a result of $\phi$ being increasing in its second argument). This means that every convergent subsequence of $(\gamma(\lambda_n))_n$ has the same limit. As $(\gamma(\lambda_n))_n$ is bounded, it admits a limit itself (recall that if all the convergent subsequences of a bounded sequence converge to the same real limit, then the sequence itself also converges to the same limit as well).

(iii) Lastly, we now prove that the limit of $(\gamma(\lambda_n))_n$ is $\gamma(\bar{\lambda})$. Let $\lim_n \gamma(\lambda_n)$ be denoted as $\bar{\gamma}$, and recall that $\phi$ is continuous in both arguments. The sequence $(\phi(\lambda_n, \gamma(\lambda_n)))_n$ is constant by construction, as it always equals $a-b$, and thus it converges to $a-b$. Its limit equals $\phi(\bar{\lambda},\bar{\gamma})$ due to the continuity of $\phi$ is continuous. As a consequence of Lemma \ref{lemma:lambda interval}, there exists a unique $\gamma(\bar{\lambda})$ such that $\phi(\bar{\lambda}, \gamma(\bar{\lambda}))=a-b$. From this it follows that $\bar{\gamma}=\gamma(\bar{\lambda})$.
\end{proof}

Now, as a consequence of the lemmas above, Theorem \ref{th:general} can be proven.

\begin{proof}[Proof of Theorem \ref{th:general}]
Consider the interval $[\lambda_b, \lambda_a]$ as per Lemma \ref{lemma:lambda interval}. It follows that there exists a unique $\gamma_b$ such that $\ask{\lambda_b,\gamma_b}-\bid{\lambda_b,\gamma_b}=a-b$. As $\bid{\lambda_b,\gamma_b}<\PV{\lambda_b}=b$, we obtain that $b<\ask{\lambda_b,\gamma_b}<a$. 
	
Analogously, consider $\lambda_a$. There exists a unique $\gamma_a$ such that $\ask{Y(\lambda_a),\gamma_a}-\bid{Y(\lambda_a),\gamma_a}=a-b$. Because $a=\PV{Y(\lambda_a)}<\ask{Y(\lambda_a),\gamma_a}$, it results that $b<\bid{Y(\lambda_a),\gamma_a}<a$.

The functions $\ask{Y(\lambda), \gamma}$, $\bid{Y(\lambda), \gamma}$ and $\gamma(\lambda)$ are continuous in $\lambda$ (the latter statement is a consequence of Lemma \ref{lemma:continuity gamma}). Therefore, there exists $\bar{\lambda}\in (\lambda_b, \lambda_a)$ and a corresponding $\bar{\gamma}$ such that $\ask{Y(\bar{\lambda}),\bar{\gamma}}=a$ and such that $\bid{Y(\bar{\lambda}),\bar{\gamma}}=b$. From Lemma \ref{lemma:lambda interval} the pair $(\bar{\lambda},\bar{\gamma})$ satisfying (\ref{eq:system}) is unique.
\end{proof}

\section{Derivation of the Bachelier option pricing formulae\label{sec:derivation}}
Equation \eqref{eq:b} describes the dynamics of an Ornstein-Uhlenbeck process. Thus, \eqref{eq:b} admits a solution which, at time $T$, is of the form
\begin{equation}\label{eq:solution}
X_T=e^{(r-\alpha) T} X_0 + \sigma\int_0^T e^{(r-\alpha)(T-t)}\, dW_t.
\end{equation}
From \eqref{eq:solution} it follows that, under $\mathbb{Q}$, $X_T$ is normally distributed and, further, its mean and variance can be calculated analytically. More precisely, we have that its mean equals
\begin{equation}\label{eq:mean}
\mubar \coloneqq e^{(r-\alpha) T} X_0,
\end{equation}
while its variance
\begin{equation}\label{eq:variance}
\sigmabar^2 \coloneqq \frac{\sigma^2(e^{2 (r-\alpha) T} -1)}{2(r-\alpha)}.
\end{equation}

The present value of an European call option written on $X$, with maturity $T$ and strike price $K$, is given by 
\begin{align}
\callBarg{\mubar} &= e^{-r T} \mathbb{E}^{\mathbb{Q}}\left( \left(X_T-K\right)^+ \right)\nonumber\\
&=e^{-r T} \int_K^{+\infty}\frac{x-K}{\sigmabar\sqrt{2\pi}}e^{-\frac{1}{2}\left(\frac{x-\mubar}{\sigmabar}\right)^2}\,dx\nonumber\\
&= e^{-r T}\left[ (\mubar-K) \Phi\left(\frac{\mubar-K}{\sigmabar} \right) +\sigmabar\phi\left(\frac{\mubar-K}{\sigmabar} \right)\right].\label{eq:ee}
\end{align}

One can then retrieve the value of the put option in a similar manner:
\begin{align}
\puttBarg{\mubar} &= e^{-r T} \mathbb{E}^{\mathbb{Q}}\left( \left(K-X_T\right)^+ \right)\nonumber\\
&=e^{-r T} \int_{-\infty}^{K}\frac{K-x}{\sigmabar\sqrt{2\pi}}e^{-\frac{1}{2}\left(\frac{x-\mubar}{\sigmabar}\right)^2}\,dx\nonumber\\
&=e^{-rT}\left[(K-\mubar)\Phi\left(\frac{K-\mubar}{\sigmabar}\right)+\sigmabar\phi\left(\frac{K-\mubar}{\sigmabar}\right)\right].\label{eq:tt}
\end{align}
Observe that \eqref{eq:ee} and \eqref{eq:tt} satisfy the generalized put-call parity relationship
\begin{equation}\nonumber
e^{-\alpha T}X_0 -e^{-rT}K + \mathcal{P}-\mathcal{C}=0,
\end{equation}
see \cite[Sec. 1.2.1]{optionFormulae}.

\section{A remark on a property of the Wang transform\label{sec:remark}}
As noted in \cite{wang}, transforming a (log)normal random variable via the Wang transform allows to still obtain a (log)normal random. However, we point out that this property still holds every time a normal random variable is transformed by means of a non-decreasing and left-continuous function. To show this, assume that a random variable $X$ is normally distributed with mean $\mu$ and variance $\sigma^2$ with respect to a probability measure $\mathbb{P}$. Further, let $f:\mathbb{R}\rightarrow \mathbb{R}$ be a non-decreasing function. Denote with $f^{-1}(\,\cdot\,)$ its inverse should $f(\,\cdot\,)$ be strictly increasing, or its pseudo-inverse otherwise. In the latter case this means, given $y\in\mathbb{R}$, that
\begin{equation}\label{eq:pseudo}
f^{-1}(y)\coloneqq\inf\left\{ x\in \mathbb{R}:f(x)> y \right\}.
\end{equation}
Let $Z$ be a standard normal random variable. We obtain that\footnote{We recall the following property arising from \eqref{eq:pseudo}: given $y$ and $z$ in $\mathbb{R}$, $f(z)\leq y$ if and only if $z\leq f^{-1}(y)$. First of all, observe that $f^{-1}(y)$ in \eqref{eq:pseudo} can be rewritten as $\sup\left\{ x\in \mathbb{R}:f(x)\leq y \right\}$. Therefore, if $f(z)\leq y$, then $z\in\left\{ x\in \mathbb{R}:f(x)\leq y \right\}$, from which $z\leq f^{-1}(y)$. On the other hand, if $z\leq f^{-1}(y)$, there exists an increasing sequence $(x_n)_n$ in $\left\{ x\in \mathbb{R}:f(x)\leq y \right\}$ such that $\lim_n x_n=f^{-1}(y)$. Therefore, for every $n$, $f(x_n)\leq y$. For the left-continuity of $f$ it follows that $f(f^{-1}(y))=\lim_n f(x_n)\leq y$. As $z\leq f^{-1}(y)$, given that $f$ is non-decreasing it results that $f(z)\leq y$.}
\begin{equation}
\mathbb{P}\left( f(X) \leq u\right) = \mathbb{P}\left( f(\mu + \sigma Z)\leq u\right)= \mathbb{P}\left( Z\leq \frac{f^{-1}(u)-\mu}{\sigma}\right)=\Phi\left( \frac{f^{-1}(u)-\mu}{\sigma}\right).\label{eq:1}
\end{equation}
Therefore, by applying the Wang transform with distortion parameter $\gamma$ to \eqref{eq:1}, which we denote with $\psi_\gamma(\,\cdot\,)$, it follows that
\begin{equation}
\psi_\gamma\left(\mathbb{P}\left( f(X) \leq u\right)\right) =\Phi\left( \frac{f^{-1}(u)-\mu +\gamma\sigma}{\sigma}\right).\label{eq:2}
\end{equation}
Thus, \eqref{eq:2} shows that the distribution of $X$ is invariant with respect to the Wang transform, up to upgrading its mean from $\mu$ to $\mu-\gamma\sigma$.

\end{doublespace}

\end{document}